\documentclass{vutinfth}
\usepackage[T1]{fontenc}
\usepackage{amsmath}
\usepackage{amsthm}
\usepackage{amssymb}
\usepackage{latexsym}
\usepackage{bussproofs}
\usepackage{upgreek}
\usepackage{enumerate}
\usepackage{rotating}
\usepackage{graphicx}

\usepackage{libertine}
\usepackage{inconsolata}
\usepackage[style=alphabetic,backend=bibtex]{biblatex}
\usepackage{hyperref}
\usepackage{cleveref}

\newcommand{\haemuno}{$\mathsf{HA+EM}_1$ }
\newcommand{\haemeno}{$\mathsf{HA+EM}_1^{-}$ }

\newcommand{\emuno}{$\textsc{em}_1$ }
\newcommand{\emeno}{$\textsc{em}_1^{-}$ }
\newcommand{\emme}{$\textsc{em}$ }
\newcommand{\emmeno}{$\textsc{em}^{-}$ }

\newcommand{\mrk}{$\textsc{mrk}$}

\newcommand{\EM}                       { {\mathsf{EM}} }
\newcommand{\PA}                       { {\mathsf{PA}} }
\newcommand{\IL}                       { {\mathsf{IL}} }
\newcommand{\HA}                       { {\mathsf{HA}} }
\newcommand{\Language}                 {\mathcal{L}}
\newcommand{\E}[3]                   {{ #2 \parallel_{#1} #3}}

\newcommand{\inj}                   {{{\upiota}}}
\newcommand{\emp}[1]    {{\mathsf{P}_{#1}}}
\newcommand{\hyp}[3]                  {{\mathtt{H}_{#1}^{\forall {#2} \mathsf{#3}}}}
\newcommand{\Hyp}[2]                   {{\mathtt{H}^{\forall {#2} \mathsf{#1}}}}
\newcommand{\wit}[3]                 {{\mathtt{W}_{#1}^{\exists {#2} \mathsf{#3}^\bot}}}

\newcommand{\real}              {\, \Vdash\, }
\newcommand{\econt}[1]           {\mathcal{EM}[#1]}
\newcommand{\cruno}           {{\textbf{(CR1)}}}
\newcommand{\crdue}           {{\textbf{(CR2)}}}
\newcommand{\crtre}           {{\textbf{(CR3)}}}
\newcommand{\crquattro}           {{\textbf{(CR4)}}}
\newcommand{\crcinque}           {{\textbf{(CR5)}}}

\newcommand{\Nat}                      { {\texttt N} }

\newcommand{\sn}{\mathsf{SN}}
\newcommand{\rec}                          {{\mathsf{R}}}
\newcommand{\True}                     { {\texttt{True}} }
\newcommand{\False}                    { {\texttt{False}} }
\newcommand{\suc}{\mathsf{S}}

\newcommand{\pair}[2]{\langle #1,#2\rangle}

\newcommand{\nf}{\mathsf{NF}}
\newcommand{\postnf}{\mathsf{PNF}}

\newcommand{\evaluates}       { {\;\equiv\;} }

\theoremstyle{definition}
\newtheorem{definition}{Definition}
\newtheorem{lemma}{Lemma}

\theoremstyle{plain}
\newtheorem{thm}{Theorem}
\newtheorem{proposition}{Proposition}
\newtheorem*{corollary}{Corollary}

\newcommand{\authorname}{Matteo Manighetti} 
\newcommand{\thesistitle}{Computational interpretations of Markov's principle} 
\newcommand{\thesistitlede}{Berechnungsinterpretationen des Markov-Prinzips}

\hypersetup{
    pdfpagelayout   = TwoPageRight,           
    pdfauthor       = {\authorname},          
    pdftitle        = {\thesistitle},         
    pdfsubject      = {Computational interpretation of proofs},              
    pdfkeywords     = {proof theory, logic, computer science, constructive mathematics} 
}

\setsecnumdepth{subsection} 

\nonzeroparskip             

\makeindex      

\setauthor{}{\authorname}{}{male}
\setadvisor{Privatdoz. Dipl.-Ing Dr. techn.}{Stefan Hetzl}{}{male}
\setfirstassistant{Dr.}{Federico Aschieri}{}{male}

\setfirstreviewer{Pretitle}{Forename Surname}{Posttitle}{male}
\setsecondreviewer{Pretitle}{Forename Surname}{Posttitle}{male}

\setsecondadvisor{Pretitle}{Forename Surname}{Posttitle}{male}

\setaddress{Address}
\setregnumber{1528764}
\setdate{30}{09}{2016} 
\settitle{\thesistitle}{\thesistitlede} 
\setsubtitle{}{} 

\setthesis{master}
\setmasterdegree{master} 
\setcurriculum{Computational Logic}{Computational Logic} 

\begin{document}

\frontmatter 

\addtitlepage{english} 
\addstatementpage

\begin{abstract}
In this thesis we are concerned with Markov's principle, a statement that originated in the Russian school of Constructive Mathematics and stated originally that ``if it is impossible that an algorithm does not terminate, then it will terminate''. This principle has been adapted to many different contexts, and in particular we are interested in its most common version for arithmetic, which can be stated as ``given a total recursive function $f$, if it is impossible that a there is no $n$ for which $f(n)=0$, then there exists an $n$ such that $f(n)=0$''. This is in general not accepted in constructivism, where stating an existential statement requires one to be able to show at request a witness for the statement: here there is no clear way to choose such an $n$.

We introduce more in detail the context of constructive mathematics from different points of view, and we show how they are related to Markov's principle. In particular, several \emph{realizability} semantics are presented, which provide interpretations of logical systems by means of different computational concepts (mainly, recursive functions and lambda calculi). This field of research gave origin to the well known paradigm often called \emph{Curry-Howard isomorphism}, or also \emph{propositions as types}, that states a correspondence between proofs in logic and programs in computer science. Thanks to this the field of proof theory, that is the metamathematical investigations of proofs as mathematical objects, became of interest for computer science and in particular for the study of programming languages.

By using modern research on the Curry-Howard isomorphism, we will obtain a more refined interpretation of Markov's principle. We will then use this results to investigate the logical properties of systems related to the principle.
\end{abstract}

\selectlanguage{english}

\tableofcontents

\mainmatter
\chapter{Introduction}
\label{cha:introduction}

Defining a proper notion of constructive mathematics and building a constructive foundation of mathematics were two central concerns of mathematical logic in the past century. Starting with Hilbert's program, and continuing through other traditions such as Brouwer's intuitionism, numerous attempts have been made, often ending up in harsh contrasts - the most famous being the one between the two just mentioned schools.

The Russian constructive school pursued the program of Constructive Recursive Mathematics, led by the intuitions of A. A. Markov who was the first in trying to put the notion of \emph{algorithm} at the heart of a foundation of mathematics.

Although ultimately less successful in the field of constructivism, some of the ideas of Markov proved later to be fundamental in understanding proofs in fields of mathematics such as analysis.

In order to present a more modern explanation of Markov's standpoint, we will first need to present a more general overview of the context of constructive mathematics.

\section{The formalist approach}
The birth itself of the modern conception of proof theory is often associated with Hilbert's famous program. As it was stated in the \emph{Grundlagen der Geometrie}
the program posed four major problems that should be addressed in order to develop a reliable foundation for a mathematical theory:
\begin{itemize}
\item The formalization of the theory, including a choice of its basic objects, relations, and axioms.
\item The proof of the consistency of the axioms.
\item The question of the mutual independence and completeness of the axioms.
\item The decision problem: is there an automatic method for deciding truth of statements in the theory?
\end{itemize}

In this thesis, we are mainly concerned with the first two points. These underline the two main characteristics of Hilbert's thought: \emph{formalism} and \emph{finitism}. The Hilbertian formalism requires the elements of the theory to be expressed as certain statements in a formal language; the mathematical practice thus could be viewed as a manipulation of these statements, in accordance to some rules. This attracted criticism from other philosophical schools, first and foremost the intuitionist school, in that it seemed like it was removing the concept of \emph{mathematical truth}, in favour of giving rise to a mere mechanical game of symbols. However, the second main aspect of the Hilbertian standpoint further clarifies the approach also in relation to this criticism: the main feature of the axiomatic system that was to be sought was its \emph{consistency}, i.e. the inability of deriving a contradiction from the axioms; in the original plan, this crucial feature had to be proved by \emph{finitistic} means. Hilbert meant with this word that they should rely on inspectable\footnote{In German \emph{anschaulich}} evidence. Such a consistency proof was seen as something that nobody could doubt of.

Hilbert's program is tightly linked to G\"odel's famous incompleteness results. We will not enter in the debate of what incompleteness meant for the development of the program; however, it is interesting to mention that G\"odel clearly specified his views with respect to the Hilbertian finitism in his Yale lectures \cite{Goedel41}. There, he states that he regards a system as finitist if it satisfies the following points:

\begin{itemize}
\item All functions and relations that are primitive in the system are respectively computable and decidable.
\item The existential quantifier is not primitive in the system. That is, existential quantifications are only an abbreviation for an explicit construction of a witness.
\item Universal quantifications can be negated only in the sense that there exists a counterexample in the sense here defined, that is an explicit construction of a counterexample.
\end{itemize}

In particular, we will draw inspiration from the second point for our notion of constructive system:

\begin{definition}[Constructive system]
\label{def:constructive}
  We call a logical system \emph{constructive} if it satisfies the following two properties:
  \begin{description}
  \item[Disjunctive property] Whenever $A \lor B$ is provable in the system, then either $A$ is provable or $B$ is provable.
  \item[Existential property] Whenever $\exists x A(x)$ is provable in the system, then there exists a term $t$ such that $A(t)$ is provable.
  \end{description}
\end{definition}

\section{Intuitionitic logic and realizability semantics}
The intuitionistic school of L.E.J. Brouwer was probably the main opponent of the formalist approach. We mentioned before that the intuitionists accused Hilbert of reducing the mathematical practice to a game of symbol manipulation without a real meaning. Indeed, the intuitionists appealed to a much more sophisticated notion of mathematics, conceiving essentially mathematical objects as free creations of the mind of the mathematicians. The mathematical practice is then a matter of human communication. Therefore, an object exists only in the moment a mathematician can mentally construct it: how can one accept an indirect argument as a mental construction? Clearly, if we can prove the impossibility of the non existence of an object, we have no way to obtain a construction we can communicate.

\subsection*{The BHK explanation of intuitionistic truth}
The refusal of formalism made by Brouwer also prevented him from really accepting any formalization of an ``intuitionistic logic''. An explanation of the usual logical connectives from the intuitionistic point of view, and the beginning of the development of an intuitionistic logical system are due to Brouwer's student Arend Heyting; this is usually known as the Brouwer-Heyting-Kolmogorov interpretation, and provides an informal notion of an intuitionistic truth:

\begin{itemize}
\item There is no construction of $\bot$.
\item A construction of $A\land B$ consists of a construction of $A$ and a construction of $B$
\item A construction of $A \lor B$ consists of a construction of $A$ or a construction of $B$
\item A construction of $A\to B$ is a construction which transforms any construction of $A$ into a construction of $B$
\item A construction of $\exists x A(x)$ consists of an element $d$ of the domain and a construction of $A(d)$
\item A construction of $\forall x A(x)$ is a method which transforms every element $d$ of the domain into a construction of $A(d)$.
\end{itemize}

Negation is then interpreted as $\neg A := A \to \bot$. We can already see from this that the principle of excluded middle $A \lor \neg A$ is not justified under this interpretation: it expands to $A \lor (A \to \bot)$, and asks for either a proof of $A$, or a method to transform proofs of $A$ into the absurdity; but clearly we have no way to do this in general. The underivability of the excluded middle as a rule proved to be the common feature of different systems of constructive logic, and thus intuitionistic logic quickly became interesting per se, regardless of the intuitionistic standpoint in mathematics.

\subsection*{Realizability semantics}
\label{sec:realizability}
The BHK semantics we have defined in the previous paragraph allows us to draw some conclusions and obtain some initial results about systems of intuitionistic logic, such as the simple argument we have used to show that the excluded middle is not justifiable. However, one immediately notices how this semantics is voluntarily informal: the notions of construction and method that are mentioned, are left unspecified. Realizability semantics are a family of semantics that can be thought of as concrete versions of the BHK semantics, whenever we consider a specific intuitionistic theory. Historically, the first such example was the original \emph{number realizability} of Kleene for the intuitionistic system of arithmetic $\HA$ \cite{Kleene45} that used objects of recursion theory in order to give concrete meaning to the concepts of \emph{construction} and \emph{algorithm} we used previously. More formally, it states when a number $e$  realizes a formula $E$ by induction on the shape of the formula:

\begin{itemize}
\item $e$ realizes $(r = t)$, if $(r = t)$ is true.
\item $e$ realizes $(A \land B)$, if $e$ codes a pair $(f,g)$ such that $f$ realizes $A$ and $g$ realizes $B$.
\item $e$ realizes $A\lor B$, if $e$ codes a pair $(f,g)$ such that if $f = 0$ then $g$ realizes $A$, and if $f > 0$ then $g$ realizes $B$.
\item $e$ realizes $A\to B$, if, whenever $f$ realizes $A$, then the $e$-th partial recursive function is defined at $f$ and its value realizes $B$.
\item $e$ realizes $\neg A$, if no $f$ realizes $A$.
\item $e$ realizes $\forall x A(x)$, if, for every $n$, the $e$-th partial recursive function is defined at $n$ and its value realizes $A(n)$.
\item $e$ realizes $\exists x A(x)$, if $e$ codes a pair $(n,g)$ and $g$ realizes $A(n)$.
\end{itemize}
Since the objects of the domain of interpretation are numbers, we can internalize the notion we have just defined by formalizing it inside the same theory of arithmetic we are interpreting.
 A formalized realizability semantics together with a semantic soundness theorem (which is often called adequacy in this framework) allows a finer analysis of intuitionistic systems. For example, given the adequacy of Kleene semantics for a system of intuitionistic arithmetic we could conclude about constructivity of the system according to our \cref{def:constructive}: whenever $A \lor B$ is provable then by adequacy it is realizable, and therefore we will have a realizer coding either a realizer of $A$ or one of $B$; similarly whenever $\exists x A(x)$ is provable, then by adequacy it is realizable and the realizer codes some $n$ and a realizer of $A(n)$. 

 Moreover, realizability is able to tell more about the computational content of intuitionistic systems. Kleene realizers are understood as codes for a G\"odel numbering of the recursive functions, and thus can represent something that we can use in order to compute. Going further in this direction, Kreisel's \emph{modified realizability} \cite{Kreisel59} defines realizers as elements of a system of typed $\lambda$-calculus: these can be in turn very similar to statements of a modern functional programming language. We can think therefore of realizability interpretations as the link between constructive systems and computational systems.

\section{Constructive Recursive Mathematics and the controversy about Markov's principle}
Constructive recursive mathematics was developed by the Russian school of constructivism starting from the 1940s. Its main contributor was A.A. Markov \cite{Markov54}, and most of the research developments in this field happened until the 1970s.

In a fashion similar to the finitistic approach, the focus in CRM is on the fact that mathematical objects should be finitely representable. In particular, they should be representable by means of suitably defined \emph{algorithms}.

The main points of the approach of CRM are, as found in \cite{Troelstra88}
\begin{itemize}
\item The objects of mathematics are algorithms. Algorithms are meant in a mathematically precise sense in the sense that they should be presented as ``words'' in some finite alphabet of symbols.
\item Limitations due to finite memory capacity are disregarded, the length of symbol strings is unbounded (though always finite).
\item Logically compound statements not involving $\exists$, $\lor$ are understood in a direct way, but existential statements and disjunctions always have to be made explicit.
\item If it is impossible that an algorithmic computation does not terminate, we may assume that it does terminate.
\end{itemize}

The last of these points is what is commonly referred to as ``Markov's principle'', and was the main point of controversy between the intuitionists and the Russian school. Indeed, all the points that were listed fit naturally in classical recursion theory; if we think at Markov's principle in this context, it represents unbounded search: it is certain that the algorithm will halt at some point, but there is no guarantee that this will happen before the end of the universe. This was firmly disagreed by intuitionists and indeed we will see that it cannot be proven from intuitionistic logic.

\section{Natural deduction and the Curry-Howard isomorphism}
In \cref{sec:realizability} we highlighted how realizability sets a correspondence between constructive systems and models of computation. An even deeper link was noted by Haskell Curry: the rules for implication introduction and elimination of natural deduction (~\cref{fig:natded}) can be put in correspondence with the rules for abstraction and application of Church's simply typed lambda calculus.

Even though it was known from the 1940s, this correspondence was not further explored until some decades later. A reason for this delay could be found in the similar lack of success of the proof system of Natural Deduction. Introduced by Gentzen together with the immediately more popular Sequent Calculus, Natural Deduction presents inference rules in couples of \emph{introduction} and \emph{elimination} rules for every logical connective. Its other feature is that proofs are dependent on \emph{assumptions} that can be made and then discharged (represented by bracketing), thus rendering the proof independent of the previously made assumption. A system of natural deduction for intuitionstic logic is presented in \cref{fig:natded}.

\begin{figure}[!htb]
  \begin{align*}
    \AxiomC{\vdots}
    \noLine
    \UnaryInfC{$A_1$}
    \AxiomC{\vdots}
    \noLine
    \UnaryInfC{$A_2$}
    \RightLabel{$\land$-I}
    \BinaryInfC{$A_1 \land A_2$}
    \DisplayProof 
    \qquad
    \AxiomC{\vdots}
    \noLine
    \UnaryInfC{$A_1 \land A_2$}
    \RightLabel{$\land$-E$_1$}
    \UnaryInfC{$A_1$}
    \DisplayProof
    \qquad
    \AxiomC{\vdots}
    \noLine
    \UnaryInfC{$A_1 \land A_2$}
    \RightLabel{$\land$-E$_2$}
    \UnaryInfC{$A_2$}
    \DisplayProof
  \end{align*}
  \begin{align*}
    \AxiomC{$[A]$}
    \noLine
    \UnaryInfC{$B$}
    \RightLabel{$\to$-I}
    \UnaryInfC{$A \to B$}
    \DisplayProof
    \qquad
    \AxiomC{\vdots}
    \noLine
    \UnaryInfC{$A \to B$}
    \noLine
    \AxiomC{\vdots}
    \noLine
    \UnaryInfC{$A$}
    \RightLabel{$\to$-E}
    \BinaryInfC{$B$}
    \DisplayProof
  \end{align*}
  \begin{align*}
    \AxiomC{\vdots}
    \noLine
    \UnaryInfC{$A$}
    \RightLabel{$\lor$-I$_1$}
    \UnaryInfC{$A \lor B$}
    \DisplayProof  
    \qquad
    \AxiomC{\vdots}
    \noLine
    \UnaryInfC{$B$}
    \RightLabel{$\lor$-I$_2$}
    \UnaryInfC{$A \lor B$}
    \DisplayProof  
    \qquad
    \AxiomC{\vdots}
    \noLine
    \UnaryInfC{$A \lor B$}
    \AxiomC{$[A]$}
    \noLine
    \UnaryInfC{$C$}
    \AxiomC{$[B]$}
    \noLine
    \UnaryInfC{$C$}
    \RightLabel{$\lor$-E}
    \TrinaryInfC{$C$}
    \DisplayProof  
  \end{align*}
  \begin{align*}
    \AxiomC{\vdots}
    \noLine
    \UnaryInfC{$A$}
    \RightLabel{$\forall$-I ($x$ not free in the assumptions)}
    \UnaryInfC{$\forall x A$}
    \DisplayProof  
    \qquad
    \AxiomC{\vdots}
    \noLine
    \UnaryInfC{$\forall x A$}
    \RightLabel{$\forall$-E}
    \UnaryInfC{$A[t/x]$}
    \DisplayProof  
  \end{align*}
  \begin{align*}
    \AxiomC{\vdots}
    \noLine
    \UnaryInfC{$A[t/x]$}
    \RightLabel{$\exists$-I}
    \UnaryInfC{$\exists x A$}
    \DisplayProof  
    \qquad
    \AxiomC{\vdots}
    \noLine
    \UnaryInfC{$\exists x A$}
    \AxiomC{$[A]$}
    \noLine
    \UnaryInfC{$C$}
    \RightLabel{$\exists$-E ($x$ not free in $C$ and in the assumptions)}
    \BinaryInfC{$C$}
    \DisplayProof  
  \end{align*}
  \begin{align*}
    \AxiomC{$\bot$}
    \RightLabel{$\bot$-E}
    \UnaryInfC{$A$}
    \DisplayProof
  \end{align*}
\caption{Natural deduction for intuitionistic
  logic} \label{fig:natded}
\end{figure}

Sequent calculus provided a more technically convenient presentation of classical logic; moreover, Gentzen introduced it with the specific aim of proving its consistency, by means of what became to be known as Gentzen's \emph{Hauptsatz}, or cut-elimination theorem. Since we are not interested in sequent calculus, we will not talk about this theorem further. We are however interested in a somehow corresponding notion in the framework of natural deduction, which is \emph{proof normalization}. A normal proof is one where no detours appear; formally, a detour is a configuration in which an introduction rule is immediately followed by an elimination of the same connective that was introduced. Given that the two kinds of rules are one the inverse of the other, such an inference can be removed in order to make the proof more direct: an example of such procedure is shown in \cref{fig:normal}.

Normalization then is the process of removing detours from a proof, with the aim of obtaining a normal one. As we mentioned, the unavailability of a normalization theorem \footnote{In his thesis Gentzen had actually included a set of detour conversions and a proof of normalization for intuitionistic natural deduction. However this remained unknown until 2005, when a manuscript of the thesis was found. For more details see \cite{Von08}}, stating that every proof could be normalized, meant that sequent calculus became the system of choice for a long period, until Dag Prawitz finally crafted a direct normalization proof for natural deduction in 1965.

\begin{figure}
  \centering
  \begin{align*}
    \AxiomC{$[A]$}
    \noLine
    \UnaryInfC{\vdots}
    \noLine
    \UnaryInfC{$B$}
    \RightLabel{$\forall$-I}
    \UnaryInfC{$A\to B$}
    \AxiomC{\vdots}
    \noLine
    \UnaryInfC{$A$}
    \RightLabel{$\forall$-E}
    \BinaryInfC{$B$}
    \DisplayProof
    \qquad
    \leadsto
    \qquad
    \AxiomC{\vdots}
    \noLine
    \UnaryInfC{$A$}
    \noLine
    \UnaryInfC{\vdots}
    \noLine
    \UnaryInfC{$B$}
    \DisplayProof
  \end{align*}
  \caption{Normalization of a non-normal proof}
  \label{fig:normal}
\end{figure}

In his work (see for example \cite{Prawitz06}), Prawitz further clarified a key feature of the rules of natural deduction: the introduction rules can be thought of as \emph{definitional} rules that describe when one is allowed to assert a certain connective, and thus its meaning; in the same way, elimination rules can be seen as \emph{operational} rules that describe how one can use a formula depending on its main connective \footnote{This idea was already expressed by Gentzen: \emph{The introductions constitute, as it were, the ``definitions'' of the symbols concerned, and the eliminations are, in the final analysis, only consequences of this, which may be expressed something like this: At the elimination of a symbol, the formula with whose outermost symbol we are dealing may be used only ``in respect of what it means according to the introduction of that symbol''}. (\cite{Gentzen35})}.

As natural deduction started gathering  more interest, William Howard studied more in depth the relationship between deduction rules of natural deduction and typing rules of typed lambda calculus, and presented what came to be known as the Curry-Howard isomorphism \cite{Howard69}. Under this isomorphism, formulas are put in correspondence with types, hence the title of Howards's work \emph{The formulae as types notion of construction}; the correspondence stretches even further, and takes different names according to the different traditions that originated from the original work. We borrow the terminology of Wadler \cite{Wadler15} and state the full framework as:
\begin{itemize}
\item \emph{Propositions} as \emph{types}, the original intuition of Howard
\item \emph{Proofs} as \emph{programs}: since every proof tree can be made to correspond with a type derivation, we have a lambda term corresponding to the proof.
\item \emph{Simplification of proofs} as \emph{evaluation of programs}: the process of detour removal is nothing but a computation, where a complex term gets reduced in order to obtain a result of the computation.
\end{itemize}

\section{Contents of the thesis}
Modern research in the Curry-Howard tradition draws heavily from all the standpoints we briefly discussed. It stems from constructivism, and intuitionistic systems are the base for most Curry-Howard systems; it is formalist in the sense that proofs are the main object of the investigation; it is finitist in the sense that, in addition to the requirement that objects of computation should be finite, it tries to make sense of classical reasoning by these means.

We will sit in this tradition, and therefore although the main object of the discussion will be a mathematical principle, we will be interested in its computational and metalogical properties. As it was already mentioned, Markov's principle was already controversial in the debate about constructivism and foundations in the first half of the XX century: \cref{cha:intu-real-mark} will be devoted to a more in-depth accounting of the birth of realizability semantics and of the status of Markov's principle in each of them.

After that we will introduce some results in the more modern line of research of realizability and Curry Howard systems for classical logic. In \cref{cha:real-class-syst}, we shall introduce a Curry Howard system able to provide a realizability semantics for the semi-classical system of arithmetic with limited excluded middle ($\HA + \EM_1$). In \cref{cha:markovs-principle-ha} we will prove some additional results on the computational and constructive properties of $\HA + \EM_1$, and we will use them to give a new computational interpretation of Markov's principle.

Based on the intuitions of \cref{cha:markovs-principle-ha}, \cref{cha:furth-gener} will introduce a Curry Howard system for a system of full classical arithmetic and a corresponding restricted version that will be shown constructive thanks to Markov's principle.


\chapter{Intuitionistic realizability and Markov's principle}
\label{cha:intu-real-mark}
Intuitionistic logic proved to be the underlying logic for many kinds of constructive mathematics. Therefore it is often referred to as \emph{the} constructive logic. The idea of realizability semantics originated in the context of intuitionistic systems, and so where the first Curry-Howard systems: indeed it was long believed that these were the only systems that allowed a computational interpretation.

After introducing the basic ideas of intuitionistic arithmetic needed to develop the theory of realizability, this chapter will present some classical realizability results and their relation to Markov's principle.

\section{Heyting Arithmetic and Markov's principle}
Throughout the introduction we made continuous references to \emph{Arithmetic}. By this name we mean, in its broadest sense, the theory of natural numbers with the usual operations of sum and product. From the point of view of logic, although a complete axiomatization cannot exist because of G\"odel's theorems, the most common axiom system for this theory is known as Peano Arithmetic, $\mathsf{PA}$. It takes the name from Giuseppe Peano, and in its modern presentation it consists of a classical theory over the language with constant terms $0, \mathbf{s},+$ and the predicate $=$, with the axioms
\begin{itemize}
\item $\forall x (x = x)$
\item $\forall x \forall y (x = y \to  y = x)$
\item $\forall x \forall y \forall z (x = y \to y = z \to x = z)$
\item $\forall x \forall y (x = y \to \mathbf{s} x = \mathbf{s} y)$
\item $\forall x \forall y (\mathbf{s} x = \mathbf{s}y \to x = y)$
\item $\forall x (\mathsf{s}x = 0 \to \bot)$
\item $\forall x (x + 0 = x)$
\item $\forall x \forall y (x + \mathbf{s}y = \mathbf{s}(a + b))$
\item $\forall x (x \cdot 0 = 0)$
\item $\forall x \forall y (x \cdot \mathbf{s}y = (x \cdot y) + x)$
\item $\forall x (\varphi(x) \to \varphi (\mathbf{s}x)) \to \varphi (0) \to \forall x \varphi(x)$, for all formulas $\varphi$
\end{itemize}
The first four axioms define our notion of equality as an equivalence relation preserved by the successor operation. Then the following two state that the successor is a bijection between the naturals and naturals greater than zero. After them we have the definitions for addition and multiplication, and finally the induction axiom scheme.

By Heyting Arithmetic, $\HA$, we mean the intuitionistic theory of the same axioms. In this context, we formulate Markov's principle as the statement \[\neg \neg \exists x A(x) \to \exists x A(x)\] where $A$ is a quantifier-free formula. Alternatively, we can also use the following form, which is equivalent under the axioms of $\HA$: \[ \neg \forall x A (x) \to \exists x \neg A(x)\]
It was mentioned in the introduction that the intuitionists did not accept Markov's principle. In line with this, neither of the formulas we just presented can be proved in the system of Heyting's intuitionistic arithmetic; however as we are going to see realizability interpretation provide mixed answer on this.

\section{G\"odel's \emph{Dialectica} interpretation}
\label{sec:godels-dialec-interpr}
Although not usually included under the category of realizability interpretations, the functional interpretation of intuitionistic arithmetic introduced by G\"odel, commonly referred to as the \emph{Dialectica} interpretation \cite{Goedel58}, is probably the first step into this line of research. As is made explicit in the title of the series of lectures where he first introduced his ideas, \emph{In what sense is intuitionistic logic constructive} \cite{Goedel41}, G\"odel aimed at making clearer the constructive meaning of the intuitionistic logical constants. In order to do this, he proposed a system of typed recursive functionals where to interpret intuitionistic theories; this approach was in his opinion finitist, as we noted in the introduction, and therefore more suitable to develop an analysis of constructivity and consistency.

Formally, the \emph{Dialectica} interpretation assigns to every formula $F$ of $\HA$ a formula $F_D$ in a system of typed functionals that we will call \textbf{T}; $F_D$ is of the form $\exists y \forall z A(y,z,x)$, where $x,y,z$ are list of variables of arbitrary type and $A$ is quantifier free. The definition is by induction on the structure of the formula: for $A$ atomic, $A_D=A$ (identifying the symbols of the languages $\HA$ and \textbf{T}); if $F_D = \exists x \forall y A(x,y)$ and $G_D = \exists u \forall v B(u,v)$, then

\begin{itemize}
\item $(F \land G)_D = \exists x,u \forall y,v (A(x,y)\land B(u,v))$
\item $(F \lor G)_D = \exists t,x,u \forall y,v (t=0 \to A(x,y) \land t=1 \to B(u,v))$
\item $(\forall z F)_D = \exists X \forall z,y A(X(z),y,z)$
\item $(\exists z F)_D = \exists z,x \forall y A(x,y,z)$
\item  $(F \to G)_D = \exists U,Y \forall x,v A(x, Y(x,v)) \to B(U(x),v)$
\item $(\neg F)_D = \exists Y \forall x \neg A(x, Y(x))$
\end{itemize}

Note that 6 follows from 5 when defining $\neg A = A \to \bot$. If we compare this with the usual BHK semantics, which also forms the basis of other realizability semantics, we can see that it is substantially different in particular in the definition of the implication: here we find no mention of a method to transform  ``any proof'' as we had in BHK.\footnote{With regard to this, G\"odel noted: ``[the fact that one does not need to quantify over all proofs] shows that the interpretation of intuitionistic logic, in terms of computable functions, in no way presupposes Heyting's and that, moreover, it is constructive and evident in a higher degree than Heyting's. For it is exactly the elimination of such vast generalities as ``any proof'' which makes for greater evidence and constructivity.'' \cite{Gödel72}}

If one thinks of the Dialectica as a Game Semantics, its peculiarity becomes clearer: consider a game between two players, where we win if we find a term $u$ such that there is no $t$ for which $A_D(u,t)$ holds; then we have a winning strategy if we can state $\exists x \forall y \ A_D(x,y)$. The cases for the connectives different from $\to$ is quite intuitive in this framework:

\begin{itemize}
\item In the case of $A \land B$, we need to find winning strategies $x$ for $A$ and $u$ for $B$
\item In the case of $A \lor B$, we declare (depending on $t$) whether we are going to give a winning strategy $x$ for $A$ or $u$ for $B$
\item In the case of $\forall x \ A$, we need to give a winning strategy $X(z)$ for $A(z)$ for every numeral $z$ the opponent might give
\item In the case of $\exists x \ A$, we need to give a numeral $z$, together with a winning straregy for $A(z)$.
\end{itemize}

The case of implication requires more explanation. Here, the opponent gives us a strategy $x$ for $A$: note that it need not be a winning one. In order to win, we need either to provide a winning strategy for $B$, or to show that the strategy he gave us was actually not winning. From this comes the shape of the interpretation of the implication: we need to give a method $U$ to obtain a strategy for $B$ such that, whenever $v$ is a strategy that wins against $U(x)$, we can build a strategy $Y(x,v)$ that wins against $x$.

\subsubsection{Markov's principle and the \emph{Dialectica}}
The difference between the BHK semantics and the \emph{Dialectica} interpretation goes in fact much farther than this, and although one can easily check that all formulas that are provable in $\HA$ are provable in \textbf{T}, the converse is not the case. It turns out that Markov's principle is precisely one of the formulas that obtain a justification in \textbf{T} but are not provable in $\HA$. If we consider the second form of Markov's principle introduced in the previous section, we have that\\
\\
\noindent
$(\forall x A)_D = \forall x A(x)$\\
$(\neg \forall x A)_D = \exists x \ \neg A(x)$\\
\\
$(\neg A)_D = \neg A$\\
$(\exists x \neg A)_D = \exists x \ \neg A (x)$\\
\\
$(\neg \forall x A \to \exists x \neg A)_D = \exists U \forall x (\neg A(x) \to \neg A(U(x)))$


Since $\exists x \neg A(x)$ is already in the required form, it is not touched by the Dialectica. In the case of $\neg \forall x \ A$, the Dialectica interpretation of the negation states that there should be a counterexample, and asks for a functional that maps witnesses of $\forall x \ A$ (which are void in the interpretation) to counterexamples of $A$; this means that the interpretation is once again $\exists x \neg A(x)$. Therefore, since both formulas get the same interpretation, Markov's principle can be trivially interpreted.

It is interesting to note that G\"odel was aware of this result and viewed it as yet another example of the fact that intuitionistic logic was not well suited as a basic constructive logic, and the system \textbf{T} was on the other side behaving much better\footnote{``The higher degree of constructivity also appears in other facts, e.g., that Markov's principle $\neg \forall x A(x) \to \exists x \neg A(x)$ (see \cite{Kleene60}, page 157, footnote) is trivially provable for any primitive recursive $A$ and, in a more general setting, for any decidable property o of any objects $x$. This, incidentally, gives an interest to this interpretation of intuitionistic logic (no matter whether in terms of computable functions of higher types or of Turing functions) even if Heyting's logic is presupposed.'' \cite{Gödel72}}

One might now wonder how such an interpretation can be used in practice. Consider the case where we have an interpretation of the premise $\neg \forall x A$, and we want to use modus ponens together with Markov's principle to get the conclusion. We can easily see that the Dialectica interpretation validates modus ponens, as shown for example in \cite{Kohlenbach08}: assume we have the two formulas in \textbf{T}

\noindent
$\forall y \ A_D(t_1,y)$\\
$\forall x,v \ (A_D(x,t_2 (x,v)) \to B_D(t_3(x),v))$ \\

\noindent
Then we can take $t_1$ for $x$ in the second formula, and $t_2(t_1,v)$ for $y$ in the first. This results in 

\noindent
$A_D(t_1,t_2(t_1,v))$\\
$A_D(t_1,t_2 (t_1,v)) \to B_D(t_3(t_1),v)$ \\

\noindent
Therefore we have $B_D(t_3(t_1),v)$ for all $v$, and thus the functional assigned to $B$ is $t_3(t_1)$. Thus, we can view modus ponens as functional application. In our case we have

\noindent
$\neg A(t_1)$ \\
$\forall y (\neg A(y) \to \neg A (U(y)))$

\noindent
And therefore applying modus ponens results in the application $U(t_1) = t_1$, since $U =\lambda x.x$.

\section{Kleene's realizability}
Kleene was the first to investigate the notion of realizability, and indeed he was the one to introduce the word itself\footnote{As mentioned in \cite{Kleene45}, the initial development of the system is actually due to Kleene's first student David Nelson.}. Upon developing the system of recursive functions, he aimed at making the system of intuitionistic arithmetic ``more precise'', and he planned to do so by employing the system of recursive functions he contributed to formalize. More precisely, the objects of the domain of the interpretation (i.e. the realizers) are the G\"odel numbers of the recursive functions: thus Kleene's realizability is often referred to as \emph{number realizability}.

Consider the standard model of arithmetic $\mathbb{N}$ and a standard pairing function $\langle -,-\rangle: \mathbb{N}^2 \to \mathbb{N}$, together with its corresponding projection functions $\pi_1$, $\pi_2$ such that $\pi_i(\langle n_1, n_2 \rangle) = n_i$. By $\{n\}m$ we represent the result of the computation of the $n$-th partial recursive function on $m$, in a suitable model of the partial recursive functions; by $\overline{n}$ we mean the numeral (in $\HA$) representing $n$. In the classic definition of Kleene, any number $n$ is a realizer of a formula $F$ under the following cirmunstances:
\begin{description}
\item[$n \, \mathbf{r}\, s=t$] if $s=t$
\item[$n \, \mathbf{r}\, A \land B$]  if $\pi_1(n) \, \mathbf{r}\, A$ and $\pi_2(n) \, \mathbf{r}\, B$
\item[$n \, \mathbf{r}\, A \to B$] if for all $m$ such that $m \, \mathbf{r}\, A$, $\{n\}m$ is a terminating computation and $\{n\}m \, \mathbf{r}\, B$
\item[$n \, \mathbf{r}\, A \lor B$] if $\pi_1(n) = 0$ and $\pi_2(n) \, \mathbf{r}\, A$, or if $\pi_1(n) = 1$ and $\pi_2(n) \, \mathbf{r}\, B$
\item[$n \, \mathbf{r}\, \forall x A(x) $] if for all $m$, $\{n\}m$ is a terminating computation and $\{n\}m \, \mathbf{r}\, A (\overline{m})$
\item[$n \, \mathbf{r}\, \exists x A(x) $] if $\pi_1(n) \, \mathbf{r}\, A(\;\overline{ \pi_2(n)}\;)$
\end{description}


We can build a realizer for Markov's principle according to this definition. Consider the 
number $n$ such that $\{n\}m =  \langle 0, \mu i .A(i) \rangle$; here, $\mu$ denotes the usual minimization operation from the theory of partial recursive functions. This is a realizer of $\neg \neg \exists x A(x) \to \exists x A(x)$ only if whenever $m$ is a realizer of $\neg \neg \exists x A(x)$, $\{n\}m$ is a realizer of $\exists x A(x)$. Unraveling the definitions, we need $\langle 0, \mu i .A(i) \rangle$ to be a realizer of $\exists x A(x)$, i.e. $0 \, \mathbf{r}\, A(\;\overline{\mu i .A(i) }\;)$. If one assumes that $\mu i .A(i)$ does not correspond to a terminating computation, then this would mean that $\exists x A (x)$ is not realizabile; in turn, if there is no realizer of $\exists x A (x)$ then any number is a realizer of $\exists x A (x)  \to \bot \equiv \neg \exists x A (x)$; finally, since we have a realizer for $\neg \exists x A (x) \to \bot$, this would give us a realizer of $\bot$, and thus a contradiction. This ensures the termination of the computation, and therefore we have  $A(\;\overline{\mu i .A(i) }\;) \equiv \top$, and any number is a realizer of $\top$.

We can easily see the catch here: the termination of the computation is ensured by classical reasoning, and what we have done is a simple shift of the classical reasoning contained in Markov's principle to the metalevel, in this case the theory of partial recursive functions. This is of course not satisfying at all from a strictly constructive point of view.

\section{Kreisel's modified realizability}
A big step forward in the field of realizability in the direction of computer science was done by Georg Kreisel with his system of \emph{modified realizability} \cite{Kreisel59}. Kreisel's realizability differentiates itself from Kleene's by using a typed lambda calculus as the domain of interpretation. Types here are put in correspondence with formulas of $\HA$, somehow predating Howard's idea of completely identifying them by some years; moreover, the use of lambda calculus and the subsequent success of lambda calculus as the foundation for functional programming languages laid the foundation for the link between computer science and proof theory.

We begin the presentation of modified realizability by presenting the system of lambda calculus. First we need to introduce the types we are going to use:
\begin{itemize}
\item $\Nat$ is a type (intuitively, the type of naturals)
\item If $\sigma$, $\tau$ are types, then $\sigma \to \tau$, $\sigma \times \tau$, $\sigma + \tau$ are types
\end{itemize}

Then, we introduce the typed terms of the system:
\begin{itemize}
\item For every type $\sigma$, a countable set of variables $x^\sigma,y^\sigma,\dots$
\item  $0: \Nat$, $\mathbf{s}: \Nat \to \Nat$
\item For all types $\sigma$, $\rec ^\sigma : \sigma \to (\Nat \to \sigma \to \sigma) \to \Nat \to \sigma$
\item For all types $\sigma$, $\tau$, projections $\pi_1^{\sigma,\tau} : \sigma \times \tau \to \sigma$, $\pi_2^{\sigma,\tau} : \sigma \times \tau \to \tau$ and pairing $\langle -,- \rangle : \sigma \to \tau \to \sigma \times \tau$
\item If $t: \tau$, then $\lambda x^\sigma.f: \sigma \to \tau$
\item If $s: \sigma \to \tau$, $t: \sigma$, then $st : \tau$
\end{itemize}

And third, the set of reduction rules:
\begin{itemize}
\item $(\lambda x.t)s \mapsto t[s/x]$
\item $\pi_1(\langle s,t \rangle) \mapsto s$, $\pi_2(\langle s,t \rangle) \mapsto t$
\item $\rec xy0 \mapsto x$, $\rec xy(\mathbf{s}z) \mapsto yz\rec xyz$
\end{itemize}

We are now ready to define the realizability interpretation. We will not treat directly the case of $\lor$, but we will assume that $A\lor B$ is a shorthand for $\exists x ((x=0 \to A) \land (\neg(x=0) \to B))$ We do so by first assigning a type $tp(A)$ to every formula $A$:
\begin{align*}
  tp(\bot) &= tp(s=t) = \Nat & tp (A \land B) &= tp(A) \times tp(B) & tp(A \to B) &= tp(A) \to tp(B) \\
  tp(\forall x A) &= \Nat \to tp(A) & tp (\exists x A) &= \Nat \times tp(A)\\
\end{align*}

\noindent
Finally, we can state

\begin{description}
\item[$t \, \mathbf{mr}\, s=t$] if $s=t$
\item[$t \, \mathbf{mr}\, A \land B$]  if $\pi_1(t) \, \mathbf{mr}\, A$ and $\pi_2(t) \, \mathbf{mr}\, B$
\item[$t \, \mathbf{mr}\, A \to B$] if for all $s : tp(A)$, $ts \, \mathbf{mr}\, B$
\item[$t \, \mathbf{mr}\, \forall x A(x) $] if for all $m : \Nat$, $tm \, \mathbf{mr}\, A (\overline{m})$
\item[$t \, \mathbf{mr}\, \exists x A(x) $] if $\pi_1(t) \, \mathbf{mr}\, A(\;\overline{ \pi_2(t)}\;)$
\end{description}

The term calculus comes with some important properties, the main one being strong normalization. This means that every term will reduce to a normal form after a finite number of reduction steps.

If we analyze modified realizability from a game semantical point of view as we did with the Dialectica, we will notice that it only differs in the definition of the implication. Indeed, here we go back to a definition in the style of the BHK. Game semantically, here we are only talking about winning strategies: this means that when playing on the formula $A \to B$, the opponent will always give us a winning strategy for $A$, to which we should answer with a winning strategy for $B$. However, winning strategies cannot be effectively recognized, so the correctness of moves cannot be checked: this is why, when it comes to game semantics, the Dialectica represents a clearer interpretation.

The fact that Markov's principle cannot be interpreted by means of the modified realizability was already shown by Kreisel \cite{Kreisel62}, and was indeed presented as one of the main points of his system. One can argue like this: assume that Markov's principle is realizable. Then in particular, for every value of $n$ one could realize $\neg \forall x \ \mathbf{T}^\bot(n,n,x) \to \exists x \ \mathbf{T}(n,n,x)$, where \textbf{T} is Kleene's predicate and is interpreted as saying ``the Turing machine $\phi_n$ terminates the computation after $x$ steps on input $n$'' (this is known to be primitive recursive and thus representable in $\HA$).
Let then $n$ be fixed, and since we have that $tp(\neg \forall x \mathbf{T}^\bot(n,n,x)) = (\Nat \to \Nat) \to \Nat$, consider the dummy term $d := \lambda x^\Nat y^\Nat.0 : (\Nat \to \Nat) \to \Nat$. By applying the realizer of Markov's principle to this dummy term, we will get a term of type $tp(\exists x \mathbf{T}) = \Nat \times \Nat$; this last term will then normalize to a term of the form $\langle m, t \rangle$, such that $m$ is a numeral.
Distinguish two cases:

\begin{enumerate}
\item If $\mathbf{T}(n,n,m)$ holds, then we have found that the $n$th Turing machine will halt on input $n$ after $m$ steps
\item If $\mathbf{T}(n,n,m)$ does not hold, we claim that the $n$th Turing machine does not halt on input $n$. Suppose that it halts, then we would have that $\forall x \ \mathbf{T}^\bot (n,n,x)$ is false and thus not realizable; this in turn means that $\neg \forall x \ \mathbf{T}^\bot (n,n,x)$ is trivially realizable by any term, and in particular by the dummy term $d$; by the definition of realizability, the realizer for Markov's principle applied to $d$ gives a realizer for $\exists x \mathbf{T}(n,n,x)$. We have already denoted the normal form of this term as $\langle m, t \rangle$, and since it is a realizer of $\exists x \mathbf{T}(n,n,x)$ it must be the case that $t$ is a realizer of $\mathbf{T}(n,n,m)$. This means that $\mathbf{T}(n,n,m)$ holds, which is a contradiction.
\end{enumerate}


Since the term calculus is strongly normalizing, we would have described a procedure that, given any $m$, decides in finite time whether the \emph{n}th Turing machine will halt on input $n$, which is well known to be an undecidable problem.


\chapter{Realizability and classical systems}
\label{cha:real-class-syst}
The previous chapter showed how realizability can be employed as a tool to analyze the constructivity of deductive systems, and to extract computational content from proofs in these systems. Given the constructive nature of realizability semantics and the inherent non constructivity of classical logic, it would seem impossible to obtain such a semantics for systems based on classical logic. This was widely believed until the nineties, when a correspondence between control operators in programming languages and classical reasoning in proofs was discovered. In this chapter, after a brief history of this idea, we will present a related idea for realizability interpretations called \emph{Interactive realizability}. Finally, we will introduce the Curry-Howard system $\HA+\EM_1$ for intuitionistic arithmetic with classical reasoning limited to formulas of the form $\exists x \emp{}$ with $\emp{}$ atomic, together with its realizability interpretation.

\section{Exceptions and classical logic}
Though successful in establishing links between intuitionistic theories and computational mechanisms, the Curry-Howard correspondence was for a long time regarded as incompatible with classical theories. Indeed, if we try to extend to classical logic the system of natural deduction we have introduced in \cref{cha:introduction}, we need to add a rule either for the excluded middle or for the double negation elimination (i.e. \emph{reductio ad absurdum}). In the first case, we need to do a disjunction elimination without having any possibility of knowing which of the two disjuncts actually holds; in the second, we assert a formula and all we know is that its negation leads to an absurdum. It looks like we have no possibilities of recovering any computational construct.

\begin{figure}
  \begin{align*}
    \AxiomC{$[\neg(\alpha \to \beta)]$}
    \noLine
    \UnaryInfC{$\Pi_1$}
    \noLine
    \UnaryInfC{$\bot$}
    \RightLabel{$\bot_c$}
    \UnaryInfC{$\alpha \to \beta$}
    \DisplayProof
    \qquad
    \leadsto
    \qquad
    \AxiomC{$[\neg \beta]_{(1)}$}
    \AxiomC{$[\alpha \to \beta]_{(2)}$}
    \AxiomC{$[\alpha]_{(3)}$}
    \BinaryInfC{$\beta$}
    \BinaryInfC{$\bot$}
    \RightLabel{$\to$-I$_{(2)}$}
    \UnaryInfC{$\neg (\alpha \to \beta)$}
    \noLine
    \UnaryInfC{$\Pi_1$}
    \noLine
    \UnaryInfC{$\bot$}
    \RightLabel{$\bot_{c(1)}$}
    \UnaryInfC{$\beta$}
    \RightLabel{$\to$-I$_{(3)}$}
    \UnaryInfC{$\alpha \to \beta$}
    \DisplayProof
  \end{align*}
  \caption{Prawitz's normalization step for \emph{reductio ad absurdum} on an implication}
  \label{fig:prawitzabs}
\end{figure}

However, we have also mentioned that classical systems of natural deduction too are equipped with a normalization theorem. It was exactly in this observation that the solution to the riddle laid undiscovered for many years\footnote{The link between Prawitz's reductions and typing of the $\mathcal{C}$ operator was established only \emph{a posteriori}, for example in \cite{de01}}. Let's take a look at the rules that Prawitz gave for the normalization of the double negation elimination in \cref{fig:prawitzabs}: the aim is to apply the rule $\bot_c$ to a formula of lower complexity, and so one assumes the negation of the conclusion together with the entire implication and the antecedent; classical reasoning is then only applied to the negated assumption. Similar rules were given for the other logical connectives. This reduction looks very similar to the one that Felleisen gave for his $\mathcal{C}$ operator:
\[ \mathcal{C} (\lambda k. M) \to \lambda n . \mathcal{C} (\lambda k. M[k:=\lambda f .k (f n)]) \]

Presented in \cite{Felleisen88}, this operator introduced the notion of \emph{continuation}, and was the basis for the introduction of such constructs in programming languages (an example being the construct \texttt{call}/\texttt{cc} available in Scheme). It was Griffin then, who in \cite{Griffin89} proposed to type Felleisen's operator as $\neg \neg A \to A$. The idea that sequential control operators could provide a computational correspondent to classical reasoning (as opposed to pure functional flow of computation and intuitionistic reasoning) proved to be very successful, and breathed new life into the \emph{propositions as types} paradigm. Starting from ideas similar to Griffin's several other systems were developed, such as the ones from Parigot \cite{Parigot92} and Krivine \cite{Krivine09}. Generalizing to other control operators, de Groote \cite{de95} showed that mechanisms of exceptions can be put in correspondence with uses of the excluded middle.

The approach of enriching systems of lambda calculus with imperative constructs provided also a new way to approach semi-classical principle, by extending Kreisel's modified realizability with delimited control operators. Using this method, Hugo Herbelin introduced in \cite{Herbelin10} a system of intuitionistic logic with the addition of two logical rules crafted in order to correspond to \texttt{catch} and \texttt{throw} instructions for a system of delimited exceptions. 

\section{Interactive realizability}
The possibilities opened by new Curry-Howard correspondences for classical logic did not, on the other side, provoke a similar number of new systems in the field of realizability semantics. The first and major example remains the work of Krivine \cite{Krivine10}, who recently applied ideas of classical realizability to set theory in order to obtain a technique alternative to forcing. 

Interactive realizability is a new realizability semantics for classical systems introduced by Aschieri and Berardi \cite{Aschieri12,Aschieri13} based on the concept of \emph{learning}: the main idea is that realizers are programs that \emph{make hypotheses}, \emph{test} them and \emph{learn} by refuting the incorrect ones. This is obtained by means of systems of lambda calculus with exceptions mechanisms: a program will continue to execute under some assumptions, and whenever it uses an instance of an assumption, the instance gets tested; if the assumption is discovered to be false an exception is raised, and the program can continue to run using the new knowledge gained from the counterexample. Different systems of interactive realizability have been put forward for various systems such as Heyting Arithmetic with limited classical principles, or more recently full first order logic and non-classical logics.

\section{The system $\HA+\EM_{1}$}
\label{sec:system-ha+em_1}
Following the terminology of \cite{Akama04}, we will now consider the semi-classical principle $\EM_1$, that is excluded middle restricted to formulas of the form $\exists \alpha \emp{}$ with $\emp{}$ an atomic predicate\footnote{This class of formulas corresponds with the class of $\Sigma_1^0$ formulas of the arithmetical hierarchy}. System $\HA+\EM_{1}$, introduced in \cite{Aschieri13}, applies the idea of interactive realizability to an intuitionistic logic extended with this principle. We could view this as adding the axiom $\forall \alpha^\Nat \emp{} \lor \exists \alpha^\Nat \neg \emp{}$ for every atomic $\emp{}$; this however carries no useful computational meaning. The new principle is therefore treated as a disjunction elimination, where the main premise is the classical axiom and gets cut.

If we try to fit this in a Curry-Howard system, we have now two proof terms representing a construction of the same conclusion, corresponding to the two proof branches where the first and then the second disjunct are assumed. By looking at the shape of the two assumptions, we can see that in the first case we need a condition to hold for all values, while in the second we are looking for a counterexample. The idea is that we should create a new proof term where we include both possible computations, and during the computation itself we might switch from the first to the second. Hence the \emuno rule that we add to the system has the following form:

\[ \begin{array}{c} \Gamma, a: \forall \alpha^{\Nat} \emp{} \vdash u: C\ \ \ \Gamma, a: \exists \alpha^{\Nat} \neg \emp{} \vdash v:C\\
\hline
\Gamma\vdash \E{a}{u}{v} : C
\end{array} \]

Here $a$ represents a communication channel between the two possible computations. The hypothesis $\forall \alpha^{\Nat} \emp{}$ is computationally void: it only serves as a certificate for the correctness of $u$; conversely, the branch where we assume $\exists \alpha^{\Nat}\lnot \emp{}$ might ask for an actual witness in order to proceed. Informally, what we want to accomplish with the reduction rules is that we should reduce inside $u$ and check for all the used instances of the universal hypothesis whether $\emp{}[n/\alpha]$ is actually true. Whenever one such instance is refuted, we have found a witness for $\neg \emp{}$, and we can employ it for the execution of $v$. This is obtained by new terms that we should use when we introduce assumptions that are to be eliminated via classical reasoning: we introduce the two typing rules
\[\Gamma, a:{\forall \alpha^{\Nat} \emp{}}\vdash \hyp{a}{\alpha}{P}: \forall\alpha^{\Nat} \emp{}\]
\[\Gamma, a:{\exists \alpha^{\Nat} \lnot\emp{}}\vdash \mathtt{W}_{a}^{\exists \alpha \neg \mathsf{P}}: \exists\alpha^{\Nat} \lnot \emp{}\]

In the first case, we introduce a term that makes the \emph{hypothesis} that $\emp{}$ holds for all values of $\alpha$; in the second, the proof term waits for a \emph{witness} for which $\emp{}$ does not hold. From an operational point of view, terms of the form $\hyp{}{\alpha}{P}$ are the ones who can raise an exception, and terms of the form $\mathtt{W}_{a}^{\exists \alpha \neg \mathsf{P}}$ are those who will catch it.

\begin{figure*}[!htb]
\begin{description}
\item[Grammar of Untyped Terms]
  \[t,u, v::=\ x\ |\ tu\ |\ tm\ |\ \lambda x u\ |\ \lambda
    \alpha u\ |\ \langle t, u\rangle\ |\ \pi_0u\ |\ \pi_{1} u\ |\
    \inj_{0}(u)\ |\ \inj_{1}(u)\ |\ (m, t)\ | t[x.u, y.v]\ |\
    t[(\alpha, x). u]\]
  \[|\ \E{a}{u}{v}\ |\ \hyp{a}{\alpha}{P}\ |\ \wit{a}{\alpha}{P}\ |\ \True \ |\ \rec u v m \ |\
    \mathsf{r}t_{1}\ldots t_{n}\] where $m$ ranges over terms of $\Language$, $x$ over variables of the lambda calculus and $a$ over $\EM_1$ hypothesis variables. Moreover, in terms of the form $\E{a}{u}{v}$ there is a $\emp{}$ such that all the free occurrences of $a$ in $u$ are of the form $\hyp{a}{\alpha}{P}$ and those in $v$ are of the form $\wit{a}{\alpha}{P}$.
\item[Contexts] With $\Gamma$ we denote contexts of the form $e_1:A_1, \ldots, e_n:A_n$, where $e_{i}$ is either a proof-term variable $x, y, z\ldots$ or a $\EM_{1}$ hypothesis variable $a, b, \ldots$

\item[Axioms] 
$\begin{array}{c} \Gamma, x:{A}\vdash x: A 
\end{array}$
$\begin{array}{c} \Gamma, a:{\forall \alpha^{\Nat} \emp{}}\vdash \hyp{a}{\alpha}{P}: \forall\alpha^{\Nat} \emp{}
\end{array}$
$\begin{array}{c} \Gamma, a:{\exists \alpha^{\Nat} \emp{}^\bot}\vdash \wit{a}{\alpha}{P}: \exists\alpha^{\Nat} \emp{}^\bot
\end{array}$

\item[Conjunction] 
$\begin{array}{c} \Gamma \vdash u: A\ \ \ \Gamma\vdash t: B\\ \hline \Gamma\vdash \langle
u,t\rangle:
A\wedge B
\end{array}\ \ \ \ $
$\begin{array}{c} \Gamma \vdash u: A\wedge B\\ \hline\Gamma \vdash\pi_0u: A
\end{array}\ \ \ \ $
$\begin{array}{c} \Gamma \vdash u: A\wedge B\\ \hline \Gamma\vdash\pi_1 u: B
\end{array}$
\item[Implication] 
$\begin{array}{c} \Gamma\vdash t: A\rightarrow B\ \ \ \Gamma\vdash u:A \\ \hline
\Gamma\vdash tu:B
\end{array}\ \ \ \ $
$\begin{array}{c} \Gamma, x:A \vdash u: B\\ \hline \Gamma\vdash \lambda x u:
A\rightarrow B
\end{array}$
\item[Disjunction Intro.] 
$\begin{array}{c} \Gamma \vdash u: A\\ \hline \Gamma\vdash \inj_{0}(u): A\vee B
\end{array}\ \ \ \ $
$\begin{array}{c} \Gamma \vdash u: B\\ \hline \Gamma\vdash\inj_{1}(u): A\vee B
\end{array}$
\item[Disjunction Elim.] $\begin{array}{c} \Gamma\vdash u: A\vee B\ \ \ \Gamma, x: A \vdash w_1: C\ \ \ \Gamma, x:B\vdash w_2:
C\\ \hline \Gamma\vdash u [x.w_{1}, x.w_{2}]: C
\end{array}$
\item[Universal Quantification] 
$\begin{array}{c} \Gamma \vdash u:\forall \alpha^{\Nat} A\\ \hline \Gamma\vdash ut: A[t/\alpha]
\end{array} $
$\begin{array}{c} \Gamma \vdash u: A\\ \hline \Gamma\vdash \lambda \alpha u:
\forall \alpha^{\Nat} A
\end{array}$ \\
where $t$ is a term of the language $\Language$ and $\alpha$ does not occur
free in any formula $B$ occurring in $\Gamma$.

\item[Existential Quantification] 
$\begin{array}{c}\Gamma\vdash u: A[t/\alpha]\\ \hline \Gamma\vdash (
t,u):
\exists
\alpha^\Nat. A
\end{array}$ \ \ \ \
$\begin{array}{c} \Gamma\vdash u: \exists \alpha^\Nat A\ \ \ \Gamma, x: A \vdash t:C\\
\hline
\Gamma\vdash u [(\alpha, x). t]: C
\end{array} $\\
where $\alpha$ is not free in $C$
nor in any formula $B$ occurring in $\Gamma$.

\item[Induction] 
$\begin{array}{c} \Gamma\vdash u: A(0)\ \ \ \Gamma\vdash v:\forall \alpha^{\Nat}.
A(\alpha)\rightarrow A(\suc(\alpha))\\ \hline \Gamma\vdash \rec uv m : A[m/\alpha] \end{array}$

\item[Post Rules] 
$\begin{array}{c} \Gamma\vdash u_1: A_1\ \Gamma\vdash u_2: A_2\ \cdots \ \Gamma\vdash u_n:
A_n\\ \hline\Gamma\vdash u: A
\end{array}$

where $A_1,A_2,\ldots,A_n,A$ are atomic formulas of $\HA$ and the rule is a Post rule for equality, for a Peano axiom or for a classical propositional
tautology or for booleans and if $n>0$, $u=\mathsf{r} u_{1}\ldots u_{n}$, otherwise $u=\True$.

\item[EM1]$\begin{array}{c} \Gamma, a: \forall \alpha^{\Nat} \emp{} \vdash w_1: C\ \ \ \Gamma, a: \exists \alpha^{\Nat} \emp{}^\bot \vdash w_2:
C\\ \hline \Gamma\vdash \E{a}{w_{1}}{w_{2}} : C
\end{array}$
\end{description}
\caption{Term Assignment Rules for $\HA+\EM_{1}$}
\label{fig:termassignment}
\end{figure*}

In \cref{fig:termassignment}, we define a system of natural deduction for $\HA+\EM_{1}$ together with a term assignment in the spirit of Curry-Howard correspondence for classical logic; for a general treatment of this kind of systems, one could refer to textbooks such as \cite{Sorensen06}. Let $\Language$ be the language of $\HA$, three distinct classes of variables appear in the proof terms: one for proof terms, denoted usually as $x, y,\ldots$; one for quantified variables of $\Language$, denoted usually as $\alpha, \beta, \ldots$; one for hypotheses bound by $\EM_{1}$, denoted usually as $a, b, \ldots$. Atomic predicates are denoted by $\emp{}, \emp{0}, \emp{1}, \ldots$; moreover, by $\emp{}^\bot$ we denote the complement predicate of $\emp{}$, and since atomic predicates are decidable in $\HA$ we have that $\emp{}^\bot \equiv \neg \emp{}$. In the term $\E{a}{u}{v} $ all the occurrences of $a$ in $u$ and $v$ are bound. We assume the usual capture-avoiding substitution for the lambda calculus, and in addition to this we add a new kind of substitution:

\begin{definition}[Witness substitution]
  Let $v$ be any term and $n$ a closed term of $\Language$. Then
  \[ v[a:=n] \]
is the term obtained replacing every occurrence of $\wit{a}{\alpha}{P}$ in $v$ by $(n,\True)$ if $\emp{}[n/\alpha] \equiv \False$, and by $(n,\hyp{a}{\alpha}{\alpha=0}\suc 0)$ otherwise
\end{definition}

Note that the reduction rules for the system in \cref{fig:F} make it clear that the second case will never actually happen; however it is needed in order to prove the normalization of the system.

\begin{figure*}[!htb]
\begin{description}
\item[Reduction Rules for $\HA$] 
\[(\lambda x. u)t\mapsto u[t/x]\qquad (\lambda \alpha. u)t\mapsto u[t/\alpha]\]
\[ \pi_{i}\pair{u_0}{u_1}\mapsto u_i, \mbox{ for i=0,1}\]
\[\inj_{i}(u)[x_{1}.t_{1}, x_{2}.t_{2}]\mapsto t_{i}[u/x_{i}], \mbox{ for i=0,1} \]
\[(n, u)[(\alpha,x).v]\mapsto v[n/\alpha][u/x], \mbox{ for each numeral $n$} \]
\[\rec u v 0 \mapsto u\]
\[\rec u v (\suc n) \mapsto v n (\rec u v n), \mbox{ for each numeral $n$} \]
\item[Permutation Rules for \emuno]
\[(\E{a} u v) w \mapsto \E{a}{uw}{vw} \]
\[\pi_{i}(\E{a} u v) \mapsto \E{a}{\pi_{i}u}{\pi_{i}v} \]
\[(\E{a} u v)[x.w_{1}, y.w_{2}] \mapsto \E{a}{u[x.w_{1}, y.w_{2}]}{v[x.w_{1}, y.w_{2}]} \]
\[(\E{a} u v)[(\alpha, x).w] \mapsto \E{a}{u[(\alpha, x).w]}{v[(\alpha, x).w]} \]
\item[Reduction Rules for \emuno]
\[\E{a} u v\mapsto u,\ \mbox{ if $a$ does not occur free in $u$ }\]
\[\E{a} u v\mapsto v[a:=n],\ \mbox{ if $\hyp{a}{\alpha}{P}n$ occurs in $u$ and $\emp{} [n/\alpha]$ is closed and $\emp{} [n/\alpha]=\False$ }\]
\[(\hyp{a}{\alpha}{P})n \mapsto \True \mbox{ if $\emp{}[n/\alpha]$ is \emph{closed} and $\emp{}[n/\alpha] \equiv \True$}\] 
\end{description}
\caption{Reduction Rules for $\HA$ + $\EM_{1}$}
\label{fig:F}
\end{figure*}

\section{Realizability interpretation of $\HA+\EM_1$}
\label{sec:real-interpr-ha+em_1}
As we anticipated, this system can be equipped with a realizability interpretation based on the ideas of interactive realizability. In order to do this, we first need to define some classes of terms:
\begin{definition}[Terms in normal form]\mbox{}
  \begin{itemize}
  \item $\sn$ is the set of strongly normalizing untyped proof terms
  \item $\nf$ is the set of normal untyped proof terms
  \item $\postnf$ is the set of the Post normal forms (intuitively,
    normal terms representing closed proof trees made only of Post
    rules whose leaves are universal hypothesis followed by an
    elimination rule), that is: $\True\in\postnf$; for every closed
    term $n$ of $\Language$, if $\hyp{a}{\alpha}{P}n\in\nf$, then
    $\hyp{a}{\alpha}{P}n\in\postnf$; if
    $t_{1}, \ldots, t_{n}\in\postnf$, then
    $\mathsf{r}t_{1}\ldots t_{n}\in\postnf$.
  \end{itemize}
\end{definition}

\begin{definition}[Quasi-Closed terms] If  $t$ is an untyped proof term which  contains as free variables only $\EM_{1}$-hypothesis variables $a_{1}, \ldots, a_{n}$,  such that each occurrence of them is of the form $\hyp{a_i}{\alpha}{P_i}$ for some $\emp{i}$, then $t$ is said to be \emph{quasi-closed}.
\end{definition}

We can now give the definition of realizers for $\HA+\EM_1$. Realizers will be quasi closed terms, and the definition will be by induction on the formula to be realized; the cases for $\land$, $\to$ and $\forall$ are the same as the ones for intuitionistic realizability we are already familiar with. The case for atomic formulas will need to be extended to take into account the case were we have open universal assumptions (since realizers are quasi-closed). Finally, the realizers for $\lor$ and $\exists$ will need a different kind of definition, with induction done also on the shape of the term.
\begin{definition}[Realizability for $\HA +\EM_{1}$]
\label{definition-reducibility}
Assume $t$ is a {\em quasi-closed} term in the grammar of untyped proof terms of $\HA+\EM_{1}$ and $C$ is a {\em closed} formula. We define the relation $t\real C$ by induction on  $C$.
\begin{enumerate}
\item
$t\real \emp{}$ if and only if  one of the following holds:
\begin{enumerate}[i)]
\item $t\in\postnf$  and  $\emp{}\evaluates \False$ implies $t$ contains a subterm  $\hyp{a}{\alpha}{Q}n$ with $\mathsf{Q}[n/\alpha]\evaluates \False$;\\

\item  $t\notin\nf$ and for all $t'$, $t\mapsto t'$ implies $t'\real \emp{}$\\
\end{enumerate}
\item
$t\real {A\wedge B}$ if and only if $\pi_0t \real {A}$ and  $\pi_1t\real {B}$\\

\item
$t\real {A\rightarrow B}$ if and only if for all $u$, if $u\real {A}$,
then $tu\real {B}$\\

\item
$t\real {A\vee B}$  if and only if  one of the following holds:\\
\begin{enumerate}[i)]
\item $t={\inj_{0}(u)}$ and $u\real A$ or $t={\inj_{1}(u)}$ and $u\real B$;\\
\item  $t=\E{a}{u}{v}$ and  $u\real A\lor B$ and $v[a:=m]\real A\lor B$ for every numeral $m$;\\
\item  $t\notin\nf$ is neutral and for all $t'$, $t\mapsto t'$ implies $t'\real A\lor B$.\\
\end{enumerate}
\item
$t\real {\forall \alpha^{\Nat} A}$ if and only if for every closed term $n$ of $\Language$,
$t{n}\real A[{n}/\alpha]$\\
\item
$t\real \exists \alpha^{\Nat} A$ if and only if one of the following holds:

\begin{enumerate}[i)]
\item $t={(n,u)}$ for some numeral $n$ and $u \real A[{n}/\alpha]$;\\

\item  $t=\E{a}{u}{v}$ and  $u\real \exists \alpha^{\Nat}A$ and $v[a:=m]\real \exists \alpha^{\Nat}A$ for every numeral $m$;\\

\item  $t\notin\nf$ is neutral and for all $t'$, $t\mapsto t'$ implies $t'\real \exists \alpha^{\Nat}A$.\\
\end{enumerate}
\end{enumerate}
\end{definition}

As we said, realizers are quasi closed terms: this means that in general realizers could contain open universal assumptions, and thus their correctness depends on them. The base cases of the definition of the realizers for the disjunction and existential quantifiers are again the same as the ones for modified realizability; however, we add a second clause that takes into account the situation where the realizer has used some assumptions; in these cases, we ask that both parts of a term of the shape $\E{a}{u}{v}$ are realizers of the formula in their turn. In a realizer with such a shape, $u$ will then be a realizer with a new open assumption in the form of a term $\hyp{a}{\alpha}{P}$, as the ones just described; $v$, on the opposite, needs a witness in order to compute and therefore we need to substitute a witness in it in order to obtain a realizer.
What this means is that these realizers will still contain a realizer in the usual shape of the clauses (i), but in the form of a \emph{prediction}, as we will see in \cref{proposition-disj}.

We conclude the section by giving some properties of the system, as they are found in the original paper \cite{Aschieri13}, that will be employed in the rest of the thesis. First of all, we define a version of the properties of reducibility candidates in the style of Girard \cite{Girard89}

\begin{definition}  Let  $t$ be a realizer of a formula $A$, define the following properties for $t$, plus an inhabitation property \crcinque\ for $A$:

\cruno\ If $t\real A$, then $t\in \sn$.\\

\crdue\ If $t \real A$ and $t\mapsto^{*} t'$, then $t' \real A$.\\

\crtre\  If $t\notin\nf$ is neutral and  for every $t'$, $t\mapsto t'$ implies $t'\real A$, then $t \real A$.\\

\crquattro\ If $t=\E{a}{u}{v}$, $u\real A$ and $v[a:=m]\real A$ for every numeral $m$, then $t\real A$. \\

\crcinque\ There is a $u$ such that $u\real A$.
\end{definition}

\begin{proposition}\label{proposition-candidates}
Every term $t$ has the properties \cruno, \crdue, \crtre, \crquattro\ and  the inhabitation property \crcinque\ holds.
\end{proposition}
\begin{proof} By induction on $C$.
\begin{itemize}
\item $C=\emp{}$ is atomic.\\

\cruno.\ By induction on the definition of $t\real \emp{}$. If $t\in\postnf$, then $t\in\sn$. If $t\notin\nf$ is neutral, then  $t\mapsto t'$ implies $t\real \emp{}$ and thus by induction hypothesis $t'\in\sn$; so $t\in\sn$.
Suppose then $t=\E{a}{u}{v}$. Since $u\real \emp{}$ and for all numerals $n$, $v[a:=n]\real \emp{}$, we have by induction hypothesis $u\in\sn$ and for all numerals $n$, $v[a:=n]\in\sn$; but these last two conditions  are easily seen to imply $\E{a}{u}{v}\in\sn$. \\

\crdue.\  Suppose $t\real \emp{}$.  It suffices to assume that $t\mapsto t'$ and show that  $t'\real \emp{}$. We proceed by induction on the number of the occurrences of the symbol $\E{}{}{}$ in $t$.  If $t$ is neutral, since it is not the case that $t\in\postnf$, by definition of $t \real \emp{}$ we obtain $t'\real \emp{}$.
Therefore, assume  $t$ is not neutral and thus $t=\E{a}{u}{v}$, with $u\real \emp{}$ and for all numerals $n$, $v[a:=n]\real \emp{}$. If $t'=u$ or $t'=v[a:=m]$ for some numeral $m$, we obtain the thesis. If $t'=\E{a}{u'}{v}$, with $u\mapsto u'$, then by induction hypothesis, $u'\real \emp{}$. So $\E{a}{u'}{v}\real \emp{}$ by definition. If $t'=\E{a}{u}{v'}$, with $v\mapsto v'$, then for every numeral $n$, $v[a:=n]\mapsto v'[a:=n]$, and thus by induction hypothesis $v'[a:=n]\real \emp{}$. So $\E{a}{u}{v'}\real \emp{}$ by definition. \\

\crtre\ and \crquattro\ are trivially true by definition of $t\real \emp{}$. \\

\crcinque.\ We have that $\hyp{a}{\alpha}{\,\alpha=0}\suc 0\real \emp{}$.\\

\item $C=A\rightarrow B$.\\

 \cruno.\  Suppose $t\real A\rightarrow B$. By induction hypothesis \crcinque,\  there is an $u$ such that $u\real A$; therefore, $tu\real B$. By induction hypothesis \cruno,\ $tu\in\sn$ and thus $t\in\sn$.  \\

 \crdue\ and \crtre\ are proved as in \cite{Girard89}.\\

 \crquattro. ($\Rightarrow$) Suppose $\E{a}{u}{v}\real A\rightarrow B$ and let $t\real A$. Then $(\E{a}{u}{v})t\real B$ and by \crdue,\ $\E{a}{ut}{vt}\real B$. By \crquattro,\ $ut\real B$ and for all numerals $n$, $v[a:=n]t=vt[a:=n]\real B$. We conclude that $u\real A\rightarrow B$ and $v[a:=n]\real A\rightarrow B$.\\
   ($\Leftarrow$). Suppose $u\real A\rightarrow B$ and $v[a:=n]\real A\rightarrow B$ for every numeral $n$. Let $t\real A$. We show by induction on the sum of the height of the reduction trees of $u, v, t$ (they are all  in $\sn$ by \cruno) that $(\E{a}{u}{v})t\real B$. By induction hypothesis \crtre,\ it is enough to assume $(\E{a}{u}{v})t\mapsto z$ and show $z\real B$. If $z=ut$ or $v[a:=n]t$, we are done. If $z=(\E{a}{u'}{v})t$ or $z=(\E{a}{u}{v'})t$ or $(\E{a}{u}{v})t'$, with $u\mapsto u'$, $v\mapsto v'$ and $t\mapsto t'$, we obtain $z\real B$ by \crdue\ and induction hypothesis. If $z=(\E{a}{ut}{vt})$, by induction hypothesis \crquattro,\ $z\real B$.  \\


 \crcinque.\ By induction hypothesis \crcinque,\ there is a term $u$ such that $u\real B$. We want to show that $\lambda\_. u\real A\rightarrow B$. Suppose $t\real A$: we have to show that $(\lambda\_.u)t\real B$. We proceed by induction on the sum of the height of the reduction trees of $u$ and $t$ (by \cruno,\ $u, t\in\sn$). By induction hypothesis \crtre,\ it is enough to assume $(\lambda\_.u)t\mapsto z$ and show $z\real B$. If $z=u$, we are done. If $z= (\lambda\_.u')t$ or $z=(\lambda\_.u)t'$, with $u\mapsto u'\real B$ (by \crtre) and $t\mapsto t'\real B$ (by \crtre), we obtain $z\real B$ by induction hypothesis.\\
 \item $C=\forall \alpha^{\Nat} A$ or $C=A\land B$. Similar to the case $C=A\rightarrow B$.\\

 \item $C=A_{0}\lor A_{1}$.\\

  \cruno\ By induction on the definition of $t\real A_{0}\lor A_{1}$. If $t=\inj_{i}(u)$, then $u\real A_{i}$, and by induction hypothesis \cruno, $u\in\sn$; therefore, $t\in\sn$. If $t\notin\nf$ is neutral, then  $t\mapsto t'$ implies $t'\real A_{0}\lor A_{1}$ and thus $t'\in\sn$ by induction hypothesis; therefore, $t\in\sn$.
Suppose then $t=\E{a}{u}{v}$. Since $u\real A_{0}\lor A_{1}$ and for all numerals $n$, $v[a:=n]\real A_{0}\lor A_{1}$, we have by induction hypothesis $u\in\sn$ and for all numerals $n$, $v[a:=n]\in\sn$. We conclude as in the case $C=\emp{}$ that $t\in \sn$. \\

 \crdue.\ Suppose $t\real A_{0}\lor A_{1}$. It suffices to assume that $t\mapsto t'$ and show that  $t'\real A_{0}\lor A_{1}$. We proceed by induction on the 
 definition of $t\real A_{0}\lor A_{1}$. If $t=\inj_{i}(u)$, then $t'=\inj_{i}(u')$, with $u\mapsto u'$. By definition of $t\real A_{0}\lor A_{1}$, we have $u \real A_{i}$. By induction hypothesis \crdue,\ $u'\real A_{i}$ and thus $t'\real A_{0}\lor A_{1}$.   If $t\notin\nf$ is neutral, by definition of $t\real A_{0}\lor A_{1}$, we obtain that $t'\real A_{0}\lor A_{1}$. If $t=\E{a}{u}{v}$, with $u\real A_{0}\lor A_{1}$ and for all numerals $n$, $v[a:=n]\real A_{0}\lor A_{1}$. If $t'=u$ or $t'=v[a:=m]$, we are done. If $t'=\E{a}{u'}{v}$, with $u\mapsto u'$, then by induction hypothesis, $u'\real A_{0}\lor A_{1}$. So $\E{a}{u'}{v}\real A_{0}\lor A_{1}$ by definition. If $t'=\E{a}{u}{v'}$, with $v\mapsto v'$, then for every numeral $n$, $v[a:=n]\mapsto v'[a:=n]$ and thus by induction hypothesis $v'[a:=n]\real A_{0}\lor A_{1}$. So $\E{a}{u}{v'}\real A_{0}\lor A_{1}$ by definition. \\

  \crtre\ and \crquattro\ are trivial.\\

\crcinque.\ By induction hypothesis \crcinque,\ there is a term $u$ such that $u\real A_{0}$. Thus $\inj_{0}(u)\real A_{0}\lor A_{1}$. \\

 \item $C=\exists \alpha^{\Nat} A$. Similar to the case $t=A_{0}\lor A_{1}$.
\end{itemize}
\end{proof}

\noindent
This first property can be used in order to state a first result on the meaning of realizers: if we denote by $\econt{u}$ a term of the form  $\E{}{(\E{}{(\E{}{u}{v_{1}})}{v_{2}})\ldots)}{v_{n}}$ for any $n \ge 0$, then

\begin{proposition}[Weak Disjunction and Numerical Existence Properties]\label{proposition-disj}\mbox{}
\begin{enumerate}
\item Suppose $t\real A\lor B$. Then either $t\mapsto^{*} \econt{\inj_{0}(u)}$ and $u\real A$ or $t\mapsto^{*} \econt{\inj_{1}(u)}$ and $u\real B$.
\item Suppose $t\real \exists \alpha^{\Nat} A$. Then $t\mapsto^{*} \econt{(n,u)}$ for some numeral $n$ such that $u\real A[n/\alpha]$.
\end{enumerate}
\end{proposition}
\begin{proof}\mbox{}
\begin{enumerate}
\item Since $t\in\sn$ by \cruno,\ let $t'$ be such that $t\mapsto^{*} t'\in\nf$. By \crdue,\ $t'\real A\lor B$. If $t'=\inj_{0}(u)$, we are done. The only possibility left is that $t'= \E{}{\E{}{\E{}{v}{v_{1}}}{v_{2}}\ldots}{v_{n}}$, with $v$ not of the form $\E{}{w_{0}}{w_{1}}$.  By \cref{definition-reducibility}.4.(ii) we have $v\real A\lor B$, and since $v$ is normal and not of the form $\E{}{w_{0}}{w_{1}}$, by \cref{definition-reducibility}.4.(i) we have either  $v=\inj_{0}(u)$, with $u\real A$, or $v=\inj_{1}(u)$, with $u\real B$.
\item Similar to 1.
\end{enumerate}
\end{proof}

Informally, this means that a realizer of a disjunction ``contains'' a realizer of one of the disjuncts, and a realizer of an existential statement similarly contains a witness. However, these realizers might rely on universal assumptions. We can specialize this theorem in the case of simpler existential formulas:

\begin{thm}[Existential Witness Extraction]\label{theorem-extraction}
Suppose $t$ is closed, $t\real \exists \alpha^{\Nat} \emp{}$ and $t\mapsto^{*} t'\in\nf$. Then $t'=(n,u)$ for some numeral $n$ such that $\emp{}[n/\alpha]\evaluates \True$.
\end{thm}
\begin{proof}\mbox{}
 By \cref{proposition-disj}, there is some numeral $n$ such that $t'=\econt{(n,u)}$ and $u\real \emp{}[n/\alpha]$.  So $$t'= \E{a_{m}} {\E{a_{2}} {\E{a_{1}} {(n,u)}{v_{1}}} {v_{2}}\ldots} {v_{m}}$$
Since $t'$ is closed, $u$ is quasi-closed and all its free variables are among $a_{1}, a_{2},\ldots, a_{m}$. We observe that  $u$ must be closed. Otherwise,  by \cref{definition-reducibility}.1.(i) and $u\real \emp{}[n/\alpha]$ we deduce that $u\in\postnf$, and thus $u$ should contain a subterm  $\hyp{a_i}{\alpha}{Q}n$; moreover, $\mathsf{Q}[n/\alpha]\evaluates \False$ otherwise $u$ would not be normal; but then we would have either $m\neq 0$ and $t'\notin\nf$ because $t' \mapsto \E{a_{m}}{\E{a_{2}}{v_{1}[a_1:=n]}{v_{2}}\ldots}{v_{m}}$, or $m=0$ and $t'$ non-closed.  Since $u$ is closed, we obtain  $t'=(n,u)$, for otherwise $t' \mapsto \E{a_{m}}{\E{a_{2}}{(n,u)}{v_{2}}\ldots}{v_{m}}$ and $t'\notin\nf$. Since $u\real \emp{}[n/\alpha]$, by \cref{definition-reducibility}.1.(i) it must be $\emp{}[n/\alpha]\evaluates \True$.
\end{proof}

\noindent
 We now come to the main theorem, the soundness of the realizability semantics:
\begin{thm}[Adequacy Theorem]\label{AdequacyTheorem}
Suppose that $\Gamma\vdash w: A$ in
the system $\HA + \EM_1$, with
$$\Gamma=x_1: {A_1},\ldots,x_n:{A_n}, a_{1}: \exists \alpha_{1}^{\Nat} \emp{}_{1}^{\bot},\ldots, a_{m}: \exists \alpha_{m}^{\Nat} \emp{}_{m}^{\bot}, b_{1}: \forall \alpha_{1}^{\Nat}\mathsf{Q}_{1},\ldots, b_{l}:\forall \alpha_{l}^{\Nat}\mathsf{Q}_{l}$$
and that the free variables of the formulas occurring in $\Gamma $ and $A$ are among
$\alpha_1,\ldots,\alpha_k$. For all closed terms $r_1,\ldots,r_k$ of $\Language$, if there are terms $t_1, \ldots, t_n$ such that
\[\text{ for  $i=1,\ldots, n$, }t_i\real A_i[{r}_1/\alpha_1\cdots
{r}_k/\alpha_k]\]
 then
\[w[t_1/x_1\cdots
t_n/x_n\  {r}_1/\alpha_1\cdots
{r}_k/\alpha_k\ a_{1}:=i_{1}\cdots a_{m}:=i_{m}  ]\real A[{r}_1/\alpha_1\cdots
{r}_k/\alpha_k]\]
for every numerals $i_{1}, \ldots, i_{m}$.
\end{thm}

Before proving this theorem, we need an auxiliary lemma

\begin{lemma}\label{proposition-somecases}\mbox{}
\begin{enumerate}
\item If for every $t\real A$, $u[t/x]\real B$, then  $\lambda x\, u\real A\rightarrow B$.
\item If for every closed term $m$ of $\Language$, $u[m/\alpha]\real B[m/\alpha]$, then $\lambda \alpha\, u\real \forall \alpha^{\Nat} B$.
\item If $u\real A_{0}$ and $v\real A_{1}$, then $\pi_{i}\pair{u}{v}\real  A_{i}$.
\item If ${w_{0}[x_{0}.u_{0}, x_{1}.u_{1}]}\real C$ and for all numerals $n$, ${w_{1}[x_{0}.u_{0}, x_{1}.u_{1}]}[a:=n]\real C$, then $(\E{a}{w_{0}}{w_{1}})[x_{0}.u_{0}, x_{1}.u_{1}]\real C$.
\item If $t\real A_{0}\lor A_{1}$ and for every $t_{i}\real A_{i}$ it holds $u_{i}[t_{i}/x_{i}]\real C$, then $t[x_{0}.u_{0}, x_{1}.u_{1}]\real C$.
\item If $t\real \exists \alpha^{\Nat} A$ and for every term $n$ of $\Language$ and $v\real A[n/\alpha]$ it holds $u[n/\alpha][v/x]\real C$, then $t[(\alpha, x).u]\real C$.
\end{enumerate}
\end{lemma}
\begin{proof}[Proof of \cref{proposition-somecases}]
\mbox{}
\begin{enumerate}

\item As in \cite{Girard89}.
\item As in \cite{Girard89}.
\item As in \cite{Girard89}.\\
\item  We may assume $a$ does not occur in $u_{0}, u_{1}$. By hypothesis, $w_{0}[x_{0}.u_{0}, x_{1}.u_{1}]\real C$ and for every numeral $n$, $w_{1}[x_{0}.u_{0}, x_{1}.u_{1}][a:=n]\real C$. By \cruno, in order to show $\E{a}{w_{0}}{w_{1}}[x_{0}.u_{0}, x_{1}.u_{1}]\real C$, we may proceed by induction on the sum of the sizes of the reduction trees of $w_{0}, w_{1}, u_{0}, u_{1}$. By \crtre,\ it then suffices to assume that $\E{a}{w_{0}}{w_{1}}[x_{0}.u_{0}, x_{1}.u_{1}]\mapsto z$ and show $z\real C$. If $z=w_{0}[x_{0}.u_{0}, x_{1}.u_{1}]$ or $w_{1}[a:=n][x_{0}.u_{0}, x_{1}.u_{1}]$ for some numeral $n$, we are done.
If $z=\E{a}{w_{0}'}{w_{1}}[x_{0}.u_{0}, x_{1}.u_{1}]$ or $z=\E{a}{w_{0}}{w_{1}'}[x_{0}.u_{0}, x_{1}.u_{1}]$ or $z=\E{a}{w_{0}}{w_{1}}[x_{0}.u_{0}', x_{1}.u_{1}]$ or $z=\E{a}{w_{0}}{w_{1}}[x_{0}.u_{0}, x_{1}.u_{1}']$, with $w_{i}\mapsto w_{i}'$ and $u_{i}\mapsto u_{i}'$, then by \crdue\ we can apply the induction hypothesis and obtain $z\real C$.  If
$$z=\E{a}{(w_{0}[x_{0}.u_{0}, x_{1}.u_{1}])}{(w_{1}[x_{0}.u_{0}, x_{1}.u_{1}])}$$
then $z\real C$ by \crquattro.
\\

\item Suppose $t\real A_{0}\lor A_{1}$ and for every $t_{i}\real A_{i}$ it holds $u_{i}[t_{i}/x_{i}]\real C$. In order to show $t[x_{0}.u_{0}, x_{1}.u_{1}]\real C$, we reason by induction of the definition of $t\real A_{0}\lor A_{1}$.
Since by \crcinque\  there are $v_{0},v_{1}$ such that $v_{i}\real A_{i}$, we have $u_{i}[v_{i}/x_{i}]\real A_{i}$, and thus by \cruno,\ $u_{i}[v_{i}/x_{i}]\in\sn$ and $t\in\sn$. We have three cases:\\
\begin{itemize}
\item $t=\inj_{i}(u)$. Then $u\real A_{i}$. We want to show that for every $u'\real A_{i}$, $\inj_{0}(u')[x_{0}.u_{0}, x_{1}.u_{1}]\real C$. By \crtre,\ it suffices to assume that $\inj_{0}(u)[x_{0}.u_{0}, x_{1}.u_{1}]\mapsto z$ and show $z\real C$.
We reason by induction on the sum of the sizes of the reduction trees of $u, u_{0}, u_{1}$. If $z=\inj_{i}(u')[x_{0}.u_{0}, x_{1}.u_{1}]$ or $z=t[x_{0}.u_{0}', x_{1}.u_{1}]$ or $z=t[x_{0}.u_{0}, x_{1}.u_{1}']$, with $u\mapsto u'$ and $u_{i}\mapsto u_{i}'$, then by \crdue\ we can apply the induction hypothesis and obtain $z\real C$. If $z=u_{i}[u/x_{i}]$, since $u\real A_{i}$, we obtain $z\real C$.\\

\item $t=\E{a}{w_{0}}{w_{1}}$. By induction hypothesis $w_{0}[x_{0}.u_{0}, x_{1}.u_{1}]\real C$ and for all numerals $n$, $w_{1}[a:=n][x_{0}.u_{0}, x_{1}.u_{1}]\real C$.  By 4., $\E{a}{w_{0}}{w_{1}}[x_{0}.u_{0}, x_{1}.u_{1}]\real C$.\\

\item $t\notin \nf$ is neutral.  We reason by  induction on the sum of the sizes of the reduction trees of $ u_{0}, u_{1}$. By \crtre,\ it suffices to assume that $t[x_{0}.u_{0}, x_{1}.u_{1}]\mapsto z$ and show $z\real C$. If $z=t'[x_{0}.u_{0}, x_{1}.u_{1}]$, we apply the (main) induction hypothesis and obtain $z\real C$.  If $z=t[x_{0}.u_{0}', x_{1}.u_{1}]$ or $z=t[x_{0}.u_{0}, x_{1}.u_{1}']$, with $u\mapsto u'$ and $u_{i}\mapsto u_{i}'$, then by \crdue\ we can apply the induction hypothesis and obtain $z\real C$.\\ \end{itemize}

\item Analogous to 5.
\end{enumerate}
\end{proof}

\begin{proof}{Proof of the Adequacy Theorem}
\newcommand{\substitution} [1]         { {\overline{#1}} }

Notation: for any term $v$ and formula $B$, we denote
\[v[t_1/x_1\cdots t_n/x_n\ {r}_1/\alpha_1\cdots {r}_k/\alpha_k\ a_{1}:=i_{1}\cdots a_{m}:=i_{m}  ]\]
with $\substitution{v}$ and
 \[B[{r}_1/\alpha_1\cdots {r}_k/\alpha_k]\]
 with $\substitution{B}$. We proceed by induction on $w$ 
. Consider the last rule in the derivation of $\Gamma\vdash w: A$:

\begin{enumerate}

\item If it is the rule $\Gamma \vdash \hyp{b_{j}}{\alpha_{j}}{P_{j}}:  \forall\alpha_{j}^{\Nat} \emp{j}$, then $w=\hyp{b_{j}}{\alpha_{j}}{P_{j}}$ and $A= \forall\alpha_{j}^{\Nat} \emp{j}$. So $\substitution{w}=\hyp{b_{j}}{\alpha_{j}}{\substitution{P}_{j}} $. Let $n$ be any closed term of $\Language$. We must show that $\substitution{w}n\real \substitution{\emp{j}}[n/\alpha_{j}]$. We have  $\hyp{b_{j}}{\alpha_{j}}{\substitution{P}_{j}}n\in \sn$; moreover, if $\hyp{b_{j}}{\alpha_{j}}{\substitution{P}_{j}}n\mapsto z$, then $z$ is $\True$ and $\substitution{\emp{j}}[n/\alpha_{j}]\evaluates \True$, and thus $z\real \ \substitution{\emp{j}}[n/\alpha_{j}]$;
 if $\Hyp{\substitution{P}_{j}}{\alpha_{j}}n\in\nf$, then $\substitution{\emp{j}}[n/\alpha_{j}]\evaluates \False$. We conclude $\hyp{b_{j}}{\alpha_{j}}{\substitution{P}_{j}} \real \forall\alpha_{j}^{\Nat} \substitution{\emp{j}}=\substitution{A}$.\\

\item If it is the rule $ \Gamma \vdash \wit{a_{j}}{\alpha_{j}}{P_{j}}:  \exists\alpha_{j}^{\Nat} \emp{}_{j}^{\bot}$, then $w=\wit{a_{j}}{\alpha_{j}}{P_{j}}$ and $A= \exists \alpha_{j}^{\Nat} \emp{}_{j}^{\bot}$. We have two possibilities. i) $\substitution{w}=(i_{j},\True)$  and $\substitution{\emp{j}}[i_{j}/\alpha_{j}]\evaluates \False$. But this means that $\substitution{w}\real \exists \alpha_{j}^{\Nat}  \substitution{\emp{j}}^{\bot}$. ii) $\substitution{w}=(i_{j}, \hyp{a_{j}}{\alpha}{\,\alpha=0}\suc 0)$. Again, $\substitution{w}\real \exists \alpha_{j}^{\Nat} \substitution{\emp{j}}^{\bot}$.\\

\item
If it is a $\lor$-I rule, say left (the other case is symmetric), then $w=\inj_{0}(u)$, $A=B\vee C$ and $\Gamma \vdash u: B$. So, $\substitution{w}=\inj_{0}(\substitution{u})$. By induction hypothesis $\substitution{u}\real \substitution{B}$ and thus $\substitution{u}\in\sn$.  We conclude $\inj_{0}(\substitution{u}) \real \substitution{B}\lor\substitution{C}= \substitution{A}$.
\\

\item If it is a $\vee$-E rule, then
\[w=  u [x.w_1, y.w_2] \]
 and  $\Gamma \vdash u: B\vee C$, $\Gamma, x: B \vdash w_1: D$, $\Gamma, y: C \vdash w_2: D$, $A=D$.  By induction hypothesis, we have $\substitution{u}\real \substitution{B}\lor \substitution{C}$; moreover,  for every $t\real \substitution{B}$, we have $\substitution{w}_{1}[t/x]\real \substitution{B}$ and for every $t\real \substitution{C}$, we have $\substitution{w}_{2}[t/y]\real \substitution{C}$.  By \cref{proposition-somecases}, we obtain $\substitution{w}=\substitution{u} [x.\substitution{w}_1, y.\substitution{w}_2]\real \substitution{D}$.
\\
\item The cases $\exists$-I and $\exists$-E are similar respectively to $\lor$-I and $\lor$-E.\\
\item If it is the  $\forall$-E rule, then $w=ut$, $A=B[t/\alpha]$
and $\Gamma \vdash u: \forall \alpha^{\Nat} B$. So,
$\substitution{w}=\substitution{u}\substitution{t}$.  By inductive hypothesis  $\substitution{u}\real
\forall\alpha^{\Nat} \substitution{B}$ and so $\substitution{u}\substitution{t}\real \substitution{B}[\substitution{t}/\alpha]$. \\

\item
If it is the  $\forall$-I rule, then $w=\lambda \alpha u$, $A=\forall \alpha^{\Nat} B$ and $\Gamma \vdash u: B$ (with $\alpha$ not occurring free in the formulas of $\Gamma$). So, $\substitution{w}=\lambda \alpha \substitution{u}$, since we may assume $\alpha\neq \alpha_1, \ldots, \alpha_k$. Let $t$ be any closed term of $\Language$; by \cref{proposition-somecases}), it is enough to prove that $\substitution{u}[t/\alpha]\real \substitution{B}[{t}/\alpha]$, which amounts to show that the induction hypothesis can be applied to $u$. For this purpose, we observe that, since $\alpha\neq \alpha_1, \ldots, \alpha_k$, for $i=1, \ldots, n$ we have
\[t_i\real \substitution{A}_i=\substitution{A}_i[t/\alpha]\]



  \item If it is the induction rule, then $w=
\rec u v t$, $A=B(t)$, $\Gamma \vdash u: B(0)$ and $\Gamma \vdash v:
\forall \alpha^{\Nat}. B(\alpha)\rightarrow B(\suc(\alpha))$. So,
$\substitution{w}=
\rec \substitution{u}\substitution{v}l$, for some numeral $l=\substitution{t}$.

We  prove that for all numerals $n$, $\rec \substitution{u}\substitution{v} n \real \substitution{B}({n})$. By \crtre,\ it is enough to suppose that  $\rec \substitution{u}\substitution{v} n \mapsto w$ and show that $w\real \substitution{B}({n})$. By induction hypothesis $\substitution{u}\real
\substitution{B}(0)$ and $\substitution{v}{m}\real
\substitution{B}({m})\rightarrow
\substitution{B}({\suc(m)})$ for all closed terms $m$ of $\Language$.  So by \cruno,\ we can reason by induction on the sum of the sizes of reduction trees of $\substitution{u}$ and $\substitution{v}$ and the size of $m$. If $n=0$ and $w=\substitution{u}$, then we are done. If $n=\suc(m)$ and $w=\substitution{v}m(\rec \substitution{u}\substitution{v}m)$, by induction hypothesis $\rec \substitution{u}\substitution{v}m\real \substitution{B}({m})$; therefore, $w\real \substitution{B}(\suc{m})$. If $w=\rec u' \substitution{v}m$, with $\substitution{u}\mapsto u'$, by induction hypothesis $w\real\substitution{B}(m)$. We conclude the same if  $w=\rec \substitution{u} {v}'m$, with $\substitution{v}\mapsto v'$.

We thus obtain that  $\substitution{w}\real \substitution{B}(l)=\substitution{B(t)}$.
\\

\item If it is the $\EM_{1}$ rule, then $w= \E{a}{u}{v}$, $\Gamma, a: \forall \alpha^{\Nat} \emp{} \vdash u: C$ and $\Gamma, a:\exists \alpha^{\Nat}\emp{}^{\bot} \vdash v:
C$ and $A=C$. By induction hypothesis, $\substitution{u}\real \substitution{C}$ and for all numerals $m$, $\substitution{v}[a:=m]\real \substitution{C}$. By \crquattro,\ we conclude $\substitution{w}=\E{a}{\substitution{u}}{\substitution{v}}\real \substitution{C}$. \\

  \item If it is a Post rule, the case $w$ is $\True$ is trivial, so we may assume $w=\mathsf{r}t_{1}\ldots t_{n}$, $A=\emp{}$ and $\Gamma\vdash t_{1}: \emp{1}, \ldots, \Gamma \vdash t_{n}: \emp{n}$. By induction hypothesis, for $i=1,\ldots, n$, we have $\substitution{t}_{i}\real \substitution{\emp{i}}$. By \cruno,\ we can argue by induction on the size of the reduction tree of $\substitution{w}$. We have two cases. i) $\substitution{w}\in\nf$. For $i=1,\ldots, n$, by \cref{theorem-extraction}, we obtain $\substitution{t}_{i}\in\postnf$. Therefore, also $\substitution{w}\in\postnf$. Assume now $\substitution{\emp{}}\evaluates \False$. Then, for some $i$, $\substitution{\emp{i}}\evaluates \False$. Therefore, $\substitution{t}_{i}$ contains a subterm  $[a]\Hyp{Q}{\alpha}n$ with $\mathsf{Q}[n/\alpha]\evaluates \False$ and thus also $\substitution{w}$. We conclude $\substitution{w}\real \substitution{\emp{}}$. ii) $\substitution{w}\notin\nf$.   By \crtre,\ it is enough to suppose $\substitution{w}\mapsto z$ and show $z\real \substitution{\emp{}}$. We have $z=\mathsf{r}\substitution{t}_{1}\ldots \substitution{t}_{i}'\ldots \substitution{t}_{n}$, with $\substitution{t}_{i}\mapsto \substitution{t}_{i}'$, and by \crdue,\ $\substitution{t}_{i}'\real \substitution{\emp{i}}$. By induction hypothesis, $z\real \substitution{\emp{}}$.
\end{enumerate}
\end{proof}

As an easy corollary, we get strong normalization of the system

\begin{corollary}[Strong Normalization of $\HA+\EM_1$] All terms of $\HA+\EM_{1}$ are strongly normalizing.
\end{corollary}


\chapter{Markov's principle in $\HA + \EM_1$}
\label{cha:markovs-principle-ha}
In the previous chapters we have developed the basic tools for understanding Curry-Howard systems and realizability semantics. In this chapter, we will perform a deeper analysis of the system $\HA+\EM_1$, and propose a restricted version (that we will call $\HA+\EM_1^-$) that gains more properties. In particular, we will prove a subject reduction theorem and then use it in order to show that the restricted system satisfies the requirements of constructive logic: whenever we prove a disjunction we are able to prove one of the disjuncts, and whenever we prove a simply existential statement, we are able to exhibit a witness. Finally, we will show that Markov's principle is provable in this restricted system and that it has a realizer that exhibits its computational content; moreover, we will show that Markov's principle is equivalent to the restricted form of excluded middle we have introduced.

Consider the system \haemuno of \cite{Aschieri13} presented in \cref{sec:system-ha+em_1}. We modify the rule \emuno by restricting it to the case where the conclusion is of the form $\exists x C$ with $C$ an atomic formula:
 
\begin{prooftree}
  \AxiomC{$\Gamma,\forall x \emp{} \vdash \exists x C$}
  \AxiomC{$\Gamma,\exists x \emp{}^\bot \vdash \exists x C$}
  \RightLabel{\emeno}
  \BinaryInfC{$\Gamma \vdash \exists x C$}
\end{prooftree}

We call this new rule \emeno.

\section{Subject reduction for \haemeno}
The subject reduction property asserts that whenever a proof term has a certain type, and it gets reduced a certain number of times, the reduced term will have the same type. When types are taken to correspond to formulas, subject reduction gives us two very important facts:
\begin{itemize}
\item From the paradigmatic point of view, it connects the concepts of \emph{proof normalization} and \emph{computation}. Reduction rules for the proof terms are usually direct simulations of proof normalization steps. If the system does enjoy the subject reduction property, we can effectively identify these two notions.
\item From a proof-theoretic point of view, when it is added to an adequate realizability interpretation it enables one to draw conclusions on the logical system based on the behaviour of the proof terms. A crucial example of this is given in \cref{sec:disj-exist-prop}.
\end{itemize}

More formally, we can write
\begin{definition}[Subject reduction]
\label{subred}
  A system enjoys subject reduction if whenever $\Gamma \vdash M : \tau$ and $M \mapsto^* N$, then also $\Gamma \vdash N : \tau$
\end{definition}
In \cite{Aschieri13} it is mentioned that system \haemuno has the subject reduction property, however the result is not proved. Moreover, classic textbooks such as \cite{Sorensen06} only offer a full proof for simply typed systems (i.e. where the only set of rules is $\to$-I and $\to$-E). We shall now give a detailed proof for the system \haemeno.

We first need two preliminary lemmas, similar to the ones presented in \cite{Sorensen06} but extended for our new rules. The main one is the \emph{Generation Lemma}, that given a typed term will allow us to talk about the terms and types used in its type derivation. Then we will need to make sure that substitutions (both ordinary and the witness substitution we have previously defined) do not affect typing of a term.

\begin{lemma}[Generation Lemma]
  \label{lemma:gen}
  Suppose $\Gamma \vdash t: \tau$. 
  \begin{enumerate}[(i)]
  \item If $t$ is of the form $\lambda x.u$ and $x \not \in dom(\Gamma)$, then $\tau = \tau_1 \to \tau_2$ and $\Gamma, x:\tau_1 \vdash u: \tau_2$
  \item If $t$ is of the form $uv$, then $\Gamma \vdash u: \sigma \to \tau$ and $\Gamma \vdash v : \sigma$ for some $\sigma$
  \item If $t$ is of the form $\lambda \alpha.u$ and $\alpha$ is not free in $\Gamma$ then $\tau = \forall \alpha^{\Nat} \sigma$ and $\Gamma \vdash u: \sigma$ 
  \item If $t$ is of the form $um$, where $m$ is a term in $\mathcal{L}$, then $\tau = \sigma[m/\alpha]$, and $\Gamma \vdash u : \forall \alpha^{\Nat} \sigma$.
  \item If $t$ is of the form $u[x.w_1,x.w_2]$, then there are $\tau_1$,$\tau_2$ such that $\Gamma \vdash u: \tau_1 \lor \tau_2$, $\Gamma,x:\tau_1 \vdash w_1 : \tau$, $\Gamma,x:\tau_2 \vdash w_2 : \tau$

  \item If $t$ is of the form $\inj_i(u)$, then $\tau = \tau_1 \lor \tau_2$ and $\Gamma \vdash u : \tau_i$
  \item If $t$ is of the form $\pair{u}{v}$, then $\tau = \tau_1 \land \tau_2$ and $\Gamma \vdash u:\tau_1$, $\Gamma \vdash v:\tau_2$
  \item If $t$ is of the form $\pi_i(u)$, then $\Gamma \vdash u : \tau \land \sigma$ or $\Gamma \vdash u:\sigma \land \tau$ (resp. if $i=1$ or 2)
  \item If $t$ is of the form $u[(\alpha,x).v]$, where $\alpha$ is not free in $\tau$ and $\Gamma$, then there is $\sigma$ such that $\Gamma, x: \sigma \vdash v : \tau$ and $\Gamma \vdash u : \exists \alpha^{\Nat}. \sigma $
  \item If $t$ is of the form $(m,u)$, then $\tau = \exists \alpha^{\Nat}. \tau_1$ and $\Gamma \vdash u: \tau_1[m/\alpha]$
  \item If $t$ is of the form $\rec u v m$, then $\tau = \sigma(m)$, $\Gamma \vdash u : \sigma (0)$, $\Gamma \vdash v : \forall \alpha^{\Nat}.\sigma (\alpha) \to \sigma(\suc \alpha)$
  \item If $t$ is of the form $[a]\Hyp{P}{\alpha} $, then $\Gamma \vdash [a]\Hyp{P}{\alpha} : \forall \alpha^{\Nat} \emp{}$ and $\Gamma \vdash a: \forall \alpha^{\Nat} \emp{}$
  \item If $t$ is of the form $\E{a}{u}{v}$, then  $\Gamma, a: \forall \alpha^{\Nat} \emp{} \vdash u : \tau$ and $\Gamma, a: \exists \alpha^{\Nat} \emp{}^{\bot} \vdash v : \tau$. Moreover,  $\tau=\exists \alpha \emp{}$.
  \end{enumerate}
\end{lemma}
\begin{proof}
  Consider for example the case of $t = \lambda x.u$. Then since the term has a type, the type derivation must end with the $\to$-introduction rule. Then it follows that $\tau = \tau_1 \to \tau_2$ and  $\Gamma, x:\tau_1 \vdash u: \tau_2$. The other cases are similar.
\end{proof}

\begin{lemma}[Substitution preserves types]
\label{lemma:subs}
\mbox{}
\begin{enumerate}[(i)]
\item If $\Gamma \vdash u: \tau$ and $\Gamma (x) = \Gamma^{\prime} (x)$ for all $x$ free in $u$, then $\Gamma^{\prime} \vdash u : \tau$ 
\item If $\Gamma, x : \sigma \vdash u : \tau$ and $\Gamma \vdash t : \sigma$, then $\Gamma \vdash u[t/x] : \tau$
\item If $\Gamma \vdash u : \tau$, $m \in \mathcal{L}$, then $\Gamma[m/\alpha] \vdash u[m/\alpha] : \tau[m/\alpha]$ 
\end{enumerate}
\end{lemma}

\begin{proof} \mbox{}

\begin{enumerate}[(i)]
  \item By induction on the structure of $u$. The base case is straightforward.

    Consider $u$ of the form $\lambda y v$. We can rename variable $y$ in a way such that it is not free in $\Gamma \cup \Gamma^\prime$. Then, $\tau = \tau_1 \to \tau_2$ by \cref{lemma:gen} and $\Gamma, y: \tau_1 \vdash v: \tau_2 $. From the induction hypothesis, $\Gamma^\prime, y: \tau_1 \vdash v: \tau_2 $ and using an implication introduction $\Gamma^\prime \vdash v : \tau$. Other cases are analogous.
  \item By induction on the structure of $u$.

    \begin{itemize}
    \item Base case: assume $u = y$ is a variable. Then if $y=x$, $\tau=\sigma$ and $u[t/x]=t$; if $y \not = x$, then the thesis follows from the first point.

      If $u$ is a \emeno hypothesis, the thesis follows from the first point.

    \item If $u = \lambda y v$, then we can assume, by (i), that  $y \not = x$ and $y$ does not occur in $\Gamma$. By the generation lemma we have $\tau= \tau_1 \to \tau_2 $ and $\Gamma, x: \sigma, y: \tau_1  \vdash v: \tau_2$. By the induction hypothesis $\Gamma, y : \tau_1 \vdash v[t/x]: \tau_2$ and applying implication introduction $\Gamma \vdash \lambda y v[t/x]: \tau_1 \to \tau_2 = \tau$

    \item If $u = vw$, then by the generation lemma $\Gamma \vdash v: \sigma \to \tau$ and $\Gamma \vdash w : \sigma$ for some $\sigma$. Then by the induction hypothesis $\Gamma \vdash v[t/x]: \sigma \to \tau$ and $\Gamma \vdash w[t/x] : \sigma$ and applying the implication elimination rule $\Gamma \vdash v[t/x]w[t/x] : \tau$. By the definition of substitution this also means $\Gamma \vdash vw[t/x] : \tau$.

    \item If $u = \inj_i(v)$, then  by \cref{lemma:gen} $\tau = \tau_1 \lor \tau_2$ and $\Gamma \vdash v : \tau_i$. By the induction hypothesis, $\Gamma \vdash v[t/x] : \tau_i$ and using the disjunction introduction rule $\Gamma \vdash \inj_i(v[t/x]) : \tau_i$. By definition of substitution this also means $\Gamma \vdash \inj_i(v)[t/x] : \tau_i$

    \item If $u = v[y.w_1,y.w_2]$, then by \cref{lemma:gen} there are $\tau_1$,$\tau_2$ such that $\Gamma, x:\sigma \vdash v: \tau_1 \lor \tau_2$, $\Gamma,x:\sigma,y:\tau_1 \vdash w_1 : \tau$, $\Gamma,x:\sigma,y:\tau_2 \vdash w_2 : \tau$. We can apply the induction hypothesis on all these terms and get $\Gamma \vdash v[t/x]: \tau_1 \lor \tau_2$, $\Gamma,y:\tau_1 \vdash w_1[t/x] : \tau$, $\Gamma,y:\tau_2 \vdash w_2[t/x] : \tau$. Then, using the disjunction elimination rule we obtain $\Gamma \vdash v[t/x][y.w_1[t/x],y.w_2[t/x]]:\tau_1 \lor \tau_2$, which by definition of substitution is the same as $\Gamma \vdash (v[y.w_1,y.w_2])[t/x]:\tau_1 \lor \tau_2$
    \end{itemize}
    The other cases are similar.

    \item Again by induction on the structure of $u$.
      \begin{itemize}
      \item Base case: if $u=x$ is a variable, if judgement $x: \tau$ is in $\Gamma$ we have the judgement $x: \tau[m/\alpha]$ in $\Gamma [m/\alpha]$. Similarly for the cases of $\hyp{a}{\alpha}{P}$ and $\wit{a}{\alpha}{P}$ 

      \item If $u = vn$, then by \cref{lemma:gen} $\tau = \sigma[n/\beta]$ and $\Gamma \vdash v: \forall \beta^{\Nat} \sigma$. By induction hypothesis $\Gamma[m/\alpha] \vdash v[m/\alpha]: \forall \beta^{\Nat} \sigma[m/\alpha]$. If $\alpha=\beta$, then $\forall \beta^{\Nat} \sigma[m/\alpha] = \forall \alpha^{\Nat} \sigma$; by using universal elimination we have 
$\Gamma [m/\alpha] \vdash v [m/\alpha] (n[m/\alpha]) : \sigma [n [m/\alpha] /\alpha] = \sigma [n/\alpha] [m/\alpha]$

If $\alpha \not = \beta$, then note that $\forall \beta^{\Nat} \sigma[m/\alpha] = \forall \beta^{\Nat} (\sigma[m/\alpha])$, and again using universal elimination  $\Gamma[m/\alpha] \vdash v[m/\alpha](n[m/\alpha]): \sigma[n[m/\alpha]/\beta] = \sigma[n/\beta][m/\alpha]$.

\item If $u = \lambda \beta v$ then by \cref{lemma:gen} $\beta$ is not free in $\Gamma$, $\tau = \forall \beta^\Nat \sigma$ and $\Gamma \vdash v: \sigma$.

Consider first $\alpha \not = \beta$. By induction hypothesis, $\Gamma[m/\alpha] \vdash v[m/\alpha]: \sigma[m/\alpha]$. Using universal introduction (since by renaming of bound variable $\beta$ is never free in $\Gamma [m/\alpha]$) then $\Gamma[m/\alpha] \vdash \lambda \beta v[m/\alpha] : \sigma[m/\alpha][n/\beta]$, and $\sigma[m/\alpha][n/\beta] = \sigma[n/\beta][m/\alpha]$ since $\alpha \not = \beta$.

Otherwise, if $\alpha = \beta$, since $\beta$ is not free in $\Gamma, v$ and $\sigma$ the result holds vacuosly.
      \end{itemize}
\end{enumerate}
The other cases are similar.
\end{proof}

\begin{lemma}[Witness substitution preserves type]
  \label{lemma:wsubs}
  If $\Gamma \vdash u : \tau$, then  $\Gamma \vdash u[a:=n] : \tau$
\end{lemma}
\begin{proof}
  Direct consequence of \cref{lemma:subs} (ii)
\end{proof}

We are now ready to state the main result for this section:

\begin{thm}
  \label{thm:subj}
  \haemeno has the subject reduction property
\end{thm}
\begin{proof}
  Assume $\Gamma \vdash t : \tau$ and $t \mapsto_{\beta} t^{\prime}$. Proceed by structural induction on the beta reduction. 

  Reduction rules for $\HA$:
  
  \begin{itemize}
  \item $t = (\lambda x. u)v : \tau$ and $t \mapsto u[v/x]$. By the generation lemma, $\Gamma \vdash (\lambda x. u) : \sigma \to \tau$ and $\Gamma \vdash v : \sigma$ for some $\sigma$. Again by generation lemma, $\Gamma, x: \sigma \vdash u : \tau$. Therefore by \cref{lemma:subs}, $\Gamma \vdash u[v/x] : \tau$.

  \item $t = (\lambda \alpha. u)v$ and $t \mapsto u[v/\alpha]$. By the generation lemma, $\tau = \tau_1[v/\alpha]$, and $\Gamma \vdash u : \forall \alpha^{\Nat}\tau_1$. Again by Generation, $\Gamma \vdash u:\tau_1$, and by \cref{lemma:subs}, $\Gamma \vdash u[v/\alpha] : \tau_1[v/\alpha]$

\item $t = \pi_{i}\pair{u_0}{u_1}$ and $t\mapsto u_i$. Then by \cref{lemma:gen} $\Gamma \vdash \pair{u_0}{u_1} : \tau_{i} \land \tau_{1-i}$ (with $\tau_0 = \tau$), and again by \cref{lemma:gen} $\Gamma \vdash u_0 : \tau_i$ and $\Gamma \vdash u_1 : \tau_{1-i}$. Then for $i=0,1$ we have $\Gamma \vdash u_i : \tau_0=\tau$
\item $t = \inj_{i}(u)[x_{1}.t_{1}, x_{2}.t_{2}]$ and $ t \mapsto t_{i}[u/x_{i}]$. By \cref{lemma:gen} there are $\tau_1$ and $\tau_2$ such that $\Gamma \vdash \inj_{i}(u) : \tau_1 \lor \tau_2$, $\Gamma, x: \tau_1 \vdash t_1 : \tau$ and $\Gamma, x: \tau_2 \vdash t_2 : \tau$. Again by \cref{lemma:gen}, $\Gamma \vdash u: \tau_i$, and by \cref{lemma:subs} $\Gamma \vdash t_{i}[u/x_{i}] : \tau$
\item $t = (n, u)[(\alpha,x).v]$ and $t \mapsto v[n/\alpha][u/x]$, where $\alpha$ is not free in $\Gamma \cup \{t:\tau\}$. By \cref{lemma:gen}, there is a $\sigma$ such that $\Gamma, x: \sigma \vdash v : \tau$ and $\Gamma \vdash (n,u) : \exists \alpha^{\Nat}. \sigma $. Again by \cref{lemma:gen}, $\Gamma \vdash u: \sigma[n/\alpha]$. Using \cref{lemma:subs} and the fact that $\alpha$ is not free in $\Gamma$ and $\tau$, we can write $\Gamma, x: \sigma[n/\alpha] \vdash v[n/\alpha] : \tau$; finally, again by \cref{lemma:subs}, $\Gamma \vdash v[n/\alpha][u/x] : \tau$
\end{itemize}

Rules for induction

\begin{itemize}
\item $t = \rec u v 0$ and $t \mapsto u$. By \cref{lemma:gen}, $\tau = \sigma (0)$ and $\Gamma \vdash u : \sigma (0)$.
\item $t = \rec u v (\suc n)$ and $t \mapsto v n (\rec u v n)$. By \cref{lemma:gen}, $\tau = \sigma (\suc n)$, $\Gamma \vdash u: \sigma (0)$ and $\Gamma \vdash v : \forall \alpha^{\Nat}.\sigma (\alpha) \to \sigma(\suc \alpha)$. In addition, by generation lemma on the term $\rec u v n$ we have $\Gamma \vdash \rec u v n : \sigma_1(n)$ and $\Gamma \vdash u : \sigma_1(0)$. Therefore $\sigma_1 = \sigma$. Using the universal quantification rule on $v$ we get $\Gamma \vdash vn : \sigma (n) \to \sigma (\suc n)$. Using the implication elimination rule on this and $\rec u v n$, we get $\Gamma \vdash vn (\rec u v n) : \sigma(\suc n)$
\end{itemize}

  Reduction rules for \emeno (there is no difference with the case of \emuno):
  \begin{itemize}
  \item $\Gamma \vdash ([a]\Hyp{P}{\alpha}) n : \tau $ and $([a]\Hyp{P}{\alpha}) n \mapsto \True$. By the generation lemma, $\Gamma \vdash [a]\Hyp{P}{\alpha} : \forall \alpha^{\Nat} \emp{} $ and also $\Gamma \vdash [a]\Hyp{P}{\alpha} : \forall \alpha^{\Nat} \tau_1$ and $\tau = \tau_1[m/\alpha]$. Therefore $\emp{} = \tau_1$, and by the condition of the rewrite rule $\tau = \emp{}[m/\alpha] = \True$.
  \item $\Gamma \vdash \E{a}{u}{v} : \tau$ and  $\E{a}{u}{v}  \mapsto u$. Then by the generation lemma we have $\Gamma, a: \forall \alpha^\Nat P \vdash u : \tau$. But $a$ is not free in $u$ by definition of the reduction rule, and so $\Gamma \vdash u : \tau$
  \item $\Gamma \vdash \E{a}{u}{v} : \tau$ and $\E{a}{u}{v}\mapsto v[a:=n]$. From \cref{lemma:gen} $\Gamma, a : \exists \alpha^\Nat\neg P \vdash  v:\tau$. From \cref{lemma:wsubs}, $\Gamma, a : \exists \alpha^\Nat\neg P \vdash v[a:=n] :\tau$. Since there are no free occurences of $a$ in $v[a:=n]$, $\Gamma \vdash  v[a:=n] :\tau$. 
  \end{itemize}

  Permutation rules for \emeno:

  \begin{itemize}
  \item $t = (\E{a}{u}{v}) w$ and $t \mapsto \E{a}{uw}{vw}$, where $a$ does not occur free in $w$. From the generation lemma, $\Gamma \vdash \E{a}{u}{v} : \sigma \to \tau$ and $\Gamma \vdash w : \sigma$ for some $\sigma$. Again by generation, $\Gamma, a: \forall \alpha^{\Nat} \emp{} \vdash u : \sigma \to \tau$ and $\Gamma, a: \exists \alpha^{\Nat} \emp{}^{\bot} \vdash v : \sigma \to \tau$. Applying implication elimination rule to both terms, and then \emeno, we get $\Gamma \vdash \E{a}{uw}{vw} : \tau$
  \item $t=(\E{a}{u}{v})[x.w_{1}, y.w_{2}]$ and $t \mapsto \E{a}{u[x.w_{1}, y.w_{2}]}{v[x.w_{1}, y.w_{2}]}$ . From \cref{lemma:gen} there are $\tau_1$, $\tau_2$ s.t $\Gamma \vdash \E{a}{u}{v} : \tau_1 \lor \tau_2$ and $\Gamma, x: \tau_1 \vdash w_1 :\tau$, $\Gamma, x: \tau_2 \vdash w_2 :\tau$. From \cref{lemma:gen} again, $\Gamma, a: \forall \alpha^{\Nat} \emp{} \vdash u : \tau_1 \lor \tau_2$ and $\Gamma, a: \exists \alpha^{\Nat} \emp{}^{\bot} \vdash v : \tau_1 \lor \tau_2$. Using disjunction elimination on both terms, followed by \emeno, we get $\Gamma \vdash \E{a}{u[x.w_{1}, y.w_{2}]}{v[x.w_{1}, y.w_{2}]} : \tau$.
  \item Cases $\pi_{i}(\E{a}{u}{v})  \mapsto \E{a}{\pi_{i}u}{\pi_{i}v}$ and $(\E{a}{u}{v})[(\alpha, x).w] \mapsto \E{a}{u[(\alpha, x).w]}{v[(\alpha, x).w]}$ are similar to the previous points.

  \end{itemize}

\end{proof}

\section{Disjunction and existential properties}
\label{sec:disj-exist-prop}
The subject reduction theorem we have just proved ensures that, whenever we reduce a proof term with one of the reduction rules, we will obtain another proof term of the same type. This, combined with ~\cref{AdequacyTheorem} (the adequacy theorem), allows us to draw conclusions on the behaviour of the logical system based on the behaviour of the proof terms. Such tools will be employed now to prove two important constructive properties of the system \haemeno.

Let's first recall the two main theorems we have seen in \cref{sec:real-interpr-ha+em_1} (the proofs can easily adapted to the new system $\HA+\EM_1^-$.) 

\begingroup
\def\thethm{\ref{theorem-extraction}}
\begin{thm}[Existential Witness Extraction]
Suppose $t$ is closed, $t\real \exists \alpha^{\Nat} \emp{}$ and $t\mapsto^{*} t'\in\nf$. Then $t'=(n,u)$ for some numeral $n$ such that $\emp{}[n/\alpha]\evaluates \True$.
\end{thm}
\addtocounter{thm}{-1}
\endgroup

Although this theorem only talks about $\Sigma_1^0$ formulas, this is enough for the purpose of proving the constructivity of \haemeno. Indeed, this is the only kind of existential statement that we are allowed to prove with our rule. In order to use the properties of the realizers to talk about the logic system, we will need the adequacy theorem:

\begingroup
\def\thethm{\ref{AdequacyTheorem}}
\begin{thm}[Adequacy Theorem]
Suppose that $\Gamma\vdash w: A$ in
the system $\HA + \EM_1$, with
$$\Gamma=x_1: {A_1},\ldots,x_n:{A_n}, a_{1}: \exists \alpha_{1}^{\Nat} \emp{}_{1}^{\bot},\ldots, a_{m}: \exists \alpha_{m}^{\Nat} \emp{}_{m}^{\bot}, b_{1}: \forall \alpha_{1}^{\Nat}\mathsf{Q}_{1},\ldots, b_{l}:\forall \alpha_{l}^{\Nat}\mathsf{Q}_{l}$$
and that the free variables of the formulas occurring in $\Gamma $ and $A$ are among
$\alpha_1,\ldots,\alpha_k$. For all closed terms $r_1,\ldots,r_k$ of $\Language$, if there are terms $t_1, \ldots, t_n$ such that
\[\text{ for  $i=1,\ldots, n$, }t_i\real A_i[{r}_1/\alpha_1\cdots
{r}_k/\alpha_k]\]
 then
\[w[t_1/x_1\cdots
t_n/x_n\  {r}_1/\alpha_1\cdots
{r}_k/\alpha_k\ a_{1}:=i_{1}\cdots a_{m}:=i_{m}  ]\real A[{r}_1/\alpha_1\cdots
{r}_k/\alpha_k]\]
for every numerals $i_{1}, \ldots, i_{m}$.
\end{thm}
\addtocounter{thm}{-1}
\endgroup

Combining these theorems with the new subject reduction theorem, we can now state

\begin{thm}[Disjunction property]
  Suppose $\vdash t : A \lor B$ in the system \haemeno where $t$ and $A \lor B$ are closed. Then there exists a term $u$ s.t. $\vdash u : A$ or a term $v$ s.t. $\vdash v : A$ 
\end{thm}
\begin{proof}
  If $t$ is not in normal form take $t^\prime$ such that $t \mapsto^* t^\prime$, and $t^\prime$ is in normal form. By \cref{thm:subj}, $\vdash t^\prime : A \lor B$, and then by the adequacy theorem $t^\prime \real A \lor B$. Consider now the possible cases by the definition of realizer:
  \begin{itemize}
  \item If $t^\prime = \inj_i(u)$, from \cref{lemma:gen} we have that $\vdash u : A$ or $\vdash u : B$ resp. when $i=0,1$.
  \item If $t^\prime= \E{a}{u}{v}$, then by \cref{lemma:gen} we would have $\vdash t^\prime : \exists \alpha \emp{}$ for some atomic $\emp{}$, but this contradicts the fact that $\vdash t^\prime : A \lor B$; so this case cannot be possible.
  \item Since $t^\prime$ is already in normal form, the third case cannot be possible.
  \end{itemize}
\end{proof}

With a very similar argument, we have also

\begin{thm}[Existential property]
Suppose $\vdash t : \exists \alpha A$ in the system \haemeno where $t$ and $\exists \alpha A$ are closed. Then there exists a numeral $n$ and a term $u$ s.t. $\vdash u : A[n/\alpha]$
\end{thm}
\begin{proof}
By the adequacy theorem, $t \real \exists \alpha A$. Distinguish cases on the definition of the realizability relation:
\begin{itemize}

\item If $t = (n, u)$, then by the generation lemma $\vdash u : A[n/\alpha]$
\item If $t = \E{a}{u}{v}$, then by \cref{lemma:gen} $A$ is atomic. Let $t^\prime$ be such that $t \mapsto^* t^\prime \in \nf$; then, by \cref{theorem-extraction}, $t = (n, t^\prime )$. By \cref{thm:subj} $\vdash (n, t^\prime ) : \exists \alpha A$, and by \cref{lemma:gen} we have $\vdash t^\prime : A[n/\alpha]$.
\end{itemize}
\end{proof}

\section{Rule \emeno is equivalent to Markov's principle}
The fundamental reason behind the constructive analysis of the system \haemeno was its resemblance with Markov's principle. The discussion we have done so far does not depend directly on this; however, the fact that our system is indeed constructive (in the broader sense we have used so far) provides even stronger evidence that the \emeno rule should be equivalent to Markov's principle.

Consider the usual system \haemeno, and state Markov's principle as the axiom \mrk: $\neg \forall \alpha \emp{} \to \exists \alpha \emp{}^{\bot}$. This gives a proof of the axiom:

\begin{prooftree}
\AxiomC{$[\neg \forall \alpha \emp{}]_{(1)}$}

\AxiomC{$[\forall \alpha \emp{}]_{\textsc{em}_1^{-}}$}

\BinaryInfC{$\bot$}

\UnaryInfC{$\exists \alpha \emp{}^{\bot}$}

\AxiomC{$[\exists \alpha \emp{}^{\bot}]_{\textsc{em}_1^-}$}

\RightLabel{\emeno}
\BinaryInfC{$\exists \alpha \emp{}^{\bot}$}
\RightLabel{(1)}
\UnaryInfC{$\neg \forall \alpha \emp{} \to \exists \alpha \emp{}^{\bot}$}
\end{prooftree}

Conversely, consider the system $\HA$ plus the axiom \mrk. We can obtain rule \emeno as follows: assuming we have proofs
\AxiomC{$\forall \alpha \emp{}$}
\noLine
\UnaryInfC{\vdots}
\noLine
\UnaryInfC{$\exists \alpha C$}
\DisplayProof
and 
\AxiomC{$\exists \alpha \emp{}^{\bot}$}
\noLine
\UnaryInfC{\vdots}
\noLine
\UnaryInfC{$\exists \alpha C$}
\DisplayProof
build this proof of $\exists \alpha C$:

\begin{prooftree}
  \AxiomC{$[\forall \alpha C^{\bot}]_{(1)} $}

  \AxiomC{$[\forall \alpha \emp{}]_{(2)}$}
  \noLine
  \UnaryInfC{\vdots}
  \noLine
  \UnaryInfC{$\exists \alpha C$}

  \BinaryInfC{$\mathcal{D}_1$}
  \UnaryInfC{$\bot$}
  \RightLabel{(2)}
  \UnaryInfC{$\neg \forall \alpha \emp{}$}

  \AxiomC{$[\forall \alpha C^{\bot}]_{(1)}$}

  \AxiomC{$[\exists \alpha \emp{}^{\bot}]_{(3)}$}
  \noLine
  \UnaryInfC{\vdots}
  \noLine
  \UnaryInfC{$\exists \alpha C$}

  \BinaryInfC{$\mathcal{D}_1$}
  \UnaryInfC{$\bot$}
  \RightLabel{(3)}
  \UnaryInfC{$\neg \exists \alpha \emp{}^{\bot}$}
  \UnaryInfC{$\mathcal{D}_2$}
  \UnaryInfC{$\forall \alpha \emp{}$}

  \BinaryInfC{$\bot$}
  \RightLabel{(1)}
  \UnaryInfC{$\neg \forall \alpha C^{\bot}$}
  \AxiomC{\mrk}
  \noLine
  \UnaryInfC{$\neg \forall \alpha C^{\bot} \to \exists \alpha C$}
  \BinaryInfC{$\exists \alpha C$}

\end{prooftree}

Where $\mathcal{D}_1$ is given by

\begin{prooftree}
  \AxiomC{$\forall \alpha C^{\bot}$}
  \UnaryInfC{$C^\bot(\alpha)$}
  \AxiomC{$[C(\alpha)]_{\exists}$}
  \BinaryInfC{$\bot$}
  \AxiomC{$\exists \alpha C(\alpha)$}
  \RightLabel{$\exists$}
  \BinaryInfC{$\bot$}
\end{prooftree}

And $\mathcal{D}_2$ is given by

\begin{prooftree}
  \AxiomC{$\emp{}(\alpha) \lor \emp{}^\bot(\alpha)$}
  \AxiomC{$[\emp{}(\alpha)]_{\lor\mbox{-E}}$}
  \AxiomC{$\neg \exists \alpha \emp{}^{\bot}(\alpha)$}
  \AxiomC{$[\emp{}^\bot (\alpha)]_{(1)}$}
  \UnaryInfC{$\exists \alpha \emp{}^{\bot} (\alpha)$}
  \BinaryInfC{$\bot$}
  \RightLabel{(1)}
  \UnaryInfC{$\neg \emp{}^\bot(\alpha)$}
  \AxiomC{$[\emp{}^\bot(\alpha)]_{\lor\mbox{-E}}$}
  \BinaryInfC{$\bot$}
  \UnaryInfC{$\emp{}(\alpha)$}
  \RightLabel{$\lor \mbox{-E}$}
  \TrinaryInfC{$\emp{}(\alpha)$}
  \UnaryInfC{$\forall \alpha \emp{} (\alpha)$}
\end{prooftree}
Note that in the last proof we used the axiom $\emp{}(\alpha) \lor \emp{}^\bot(\alpha)$ since $\emp{}$ is atomic and thus decidable in $\HA$.

\section{A realizer for Markov's principle}
\label{sec:real-mark-princ}
Now that we have a proof tree for Markov's principle in \haemeno, we can decorate it in order to get a realizer of the principle:
\begin{prooftree}

\AxiomC{$[x : \neg \forall \alpha B]_{(2)}$}
\AxiomC{$[\hyp{a}{\alpha}{B} : \forall \alpha B]_{\textsc{em}_1^{-}}$}
\BinaryInfC{$x \hyp{a}{\alpha}{B} : \bot$}

\UnaryInfC{$\mathsf{r}x \hyp{a}{\alpha}{B} : B^\bot[0/\alpha]$}
\UnaryInfC{$(0,\mathsf{r}x \hyp{a}{\alpha}{B}) : \exists \alpha B^\bot$}

\AxiomC{$[\wit{a}{\alpha}{B} : \exists \alpha B^{\bot}]_{\textsc{em}_1^-}$}

\RightLabel{\emeno}
\BinaryInfC{$\E{a}{(0, \mathsf{r} x \hyp{a}{\alpha}{B})}{\wit{a}{\alpha}{B}}  : \exists \alpha B^{\bot}$}
\RightLabel{(1)}
\UnaryInfC{$\lambda x.(\E{a}{(0, \mathsf{r} x \hyp{a}{\alpha}{B})}{\wit{a}{\alpha}{B}}): \neg \forall \alpha B \to \exists \alpha B^{\bot}$}
\end{prooftree}

The extracted term fully exploits the properties of the system in order to get a more precise computational meaning for Markov's principle. When a realizer for $\neg \forall \alpha B$ is given, it is applied to the hypotetical term. Thus, the computation can proceed by using this assumption and reducing inside the left hand side of the proof term. At some point however, we are guaranteed that the program will use the hypotesis on a term $m$ for which $B[m/\alpha]$ does not hold. At this point, an exception is raised and we gets the witness we were waiting for.


\chapter{Further generalizations}
\label{cha:furth-gener}
In the previous section it was shown that Markov's principle is equivalent to the excluded middle restricted to $\Sigma_1^0$ formulas and $\Sigma_1^0$ conclusions. However when taking a closer look at the formal proof we gave, it can be noticed that the crucial use of Markov's principle in proving this form of the excluded middle is only done on the conclusion. The assumptions are only used in order to obtain a contradiction and then do an \emph{ex falso} reasoning. On the other side, we need $\Sigma_1^0$ assumptions in order to be able to prove Markov's principle.


Starting from this observation, we will prove that Markov's principle is equivalent to a rule allowing arbitrary excluded middle with $\Sigma_1^0$ conclusions. After introducing this new rule, we will try to use it in order to get a direct translation from proofs of classical arithmetic. We will first introduce the well established tool of negative translations and show how they succeed in embedding classical reasoning inside intuitionistic systems. Then we will introduce a new translation that, although missing the usual properties of negative translations, will be useful for our scope when coupled with a set of proof transformation rules. Thanks to these two tools, we will provide a way to transform classical proofs of simply existential formulas into proofs in $\HA + \EM_1^-$.

We will consider again the system $\HA$ of natural deduction for intuitionistic arithmetic, and the $\HA + \EM_1^-$ extension we have already studied. When referring to \emph{classical} proofs, or proofs in Peano Arithmetic $\PA$, we mean proofs in $\HA+\EM$ where we add to $\HA$ the rule of full excluded middle:

  \begin{prooftree}
    \AxiomC{$\Gamma, A \vdash C$}
    \AxiomC{$\Gamma, \neg A \vdash C$}
    \RightLabel{\emme}
    \BinaryInfC{$\Gamma \vdash C$}
  \end{prooftree}

\section{Full excluded middle with restricted conclusions}
\label{sec:full-excluded-middle}
Consider a system of natural deduction for intuitionistic arithmetic, to which we add restricted classical reasoning in the form of rule \emmeno:

  \begin{prooftree}
    \AxiomC{$\Gamma, A \vdash \exists x \emp{}$}
    \AxiomC{$\Gamma, \neg A \vdash \exists x \emp{}$}
    \RightLabel{\emmeno}
    \BinaryInfC{$\Gamma \vdash \exists x \emp{}$}
  \end{prooftree}

That is, we allow to eliminate instances of the excluded middle for arbitrary formulas $A$, but only if the conclusion is a $\Sigma_1^0$ formula.

This deduction, similar to the one of the previous chapter, gives a proof of Markov's principle by using the excluded middle rule on the formula $\exists \alpha \emp{}^\bot$

\begin{prooftree}
\AxiomC{$[\exists \alpha \emp{}^{\bot}]_{\textsc{em}^-}$}

\AxiomC{$[\neg \forall \alpha \emp{}]_{(1)}$}

\AxiomC{$[\neg \exists \alpha \emp{}^\bot]_{\textsc{em}^{-}}$}
\noLine
\UnaryInfC{$\mathcal{D}$}
\noLine
\UnaryInfC{$\forall \alpha \emp{}$}

\BinaryInfC{$\bot$}

\UnaryInfC{$\exists \alpha \emp{}^{\bot}$}

\RightLabel{\emmeno}
\BinaryInfC{$\exists \alpha \emp{}^{\bot}$}
\RightLabel{(1)}
\UnaryInfC{$\neg \forall \alpha \emp{} \to \exists \alpha \emp{}^{\bot}$}
\end{prooftree}

Where $\mathcal{D}$ is, as in the previous section,

\begin{prooftree}
  \AxiomC{$\emp{}(\alpha) \lor \emp{}^\bot(\alpha)$}
  \AxiomC{$[\emp{}(\alpha)]_{\lor\mbox{-E}}$}
  \AxiomC{$\neg \exists \alpha \emp{}^{\bot}(\alpha)$}
  \AxiomC{$[\emp{}^\bot (\alpha)]_{(1)}$}
  \UnaryInfC{$\exists \alpha \emp{}^{\bot} (\alpha)$}
  \BinaryInfC{$\bot$}
  \RightLabel{(1)}
  \UnaryInfC{$\neg \emp{}^\bot(\alpha)$}
  \AxiomC{$[\emp{}^\bot(\alpha)]_{\lor\mbox{-E}}$}
  \BinaryInfC{$\bot$}
  \UnaryInfC{$\emp{}(\alpha)$}
  \RightLabel{$\lor \mbox{-E}$}
  \TrinaryInfC{$\emp{}(\alpha)$}
  \UnaryInfC{$\forall \alpha \emp{} (\alpha)$}
\end{prooftree}

Conversely, given a system of intuitionistic arithmetic $\HA$ with Markov's principle as axiom \mrk : $\neg \forall \alpha \emp{} \to \exists \alpha \emp{}^{\bot}$  we can obtain rule \emmeno as follows: assuming we have proofs
\AxiomC{$A$}
\noLine
\UnaryInfC{\vdots}
\noLine
\UnaryInfC{$\exists \alpha \emp{}$}
\DisplayProof
and 
\AxiomC{$\neg A$}
\noLine
\UnaryInfC{\vdots}
\noLine
\UnaryInfC{$\exists \alpha \emp{}$}
\DisplayProof
build this proof of $\exists \alpha \emp{}$:

\begin{prooftree}
  \AxiomC{$[\forall \alpha \emp{}^{\bot}]_{(1)} $}

  \AxiomC{$[\neg A]_{(2)}$}
  \noLine
  \UnaryInfC{\vdots}
  \noLine
  \UnaryInfC{$\exists \alpha \emp{}$}

  \noLine
  \BinaryInfC{$\mathcal{D}$}
  \noLine
  \UnaryInfC{$\bot$}
  \RightLabel{(2)}
  \UnaryInfC{$\neg \neg A$}

  \AxiomC{$[\forall \alpha \emp{}^{\bot}]_{(1)}$}

  \AxiomC{$[A]_{(3)}$}
  \noLine
  \UnaryInfC{\vdots}
  \noLine
  \UnaryInfC{$\exists \alpha \emp{}$}

  \noLine
  \BinaryInfC{$\mathcal{D}$}
  \noLine
  \UnaryInfC{$\bot$}
  \RightLabel{(3)}
  \UnaryInfC{$\neg A$}

  \BinaryInfC{$\bot$}
  \RightLabel{(1)}
  \UnaryInfC{$\neg \forall \alpha \emp{}^{\bot}$}
  \AxiomC{\textsc{mrk}}
  \noLine
  \UnaryInfC{$\neg \forall \alpha \emp{}^{\bot} \to \exists \alpha \emp{}$}
  \BinaryInfC{$\exists \alpha \emp{}$}

\end{prooftree}

Where $\mathcal{D}$ is given by

\begin{prooftree}
  \AxiomC{$\exists \alpha \emp{}(\alpha)$}

  \AxiomC{$\forall \alpha \emp{}^{\bot}$}
  \UnaryInfC{$\emp{}^\bot(\alpha)$}
  \AxiomC{$[\emp{}(\alpha)]_{\exists}$}
  \BinaryInfC{$\bot$}
  \RightLabel{$\exists$}
  \BinaryInfC{$\bot$}
\end{prooftree}

We have now a more general result than the one we had in the previous chapter: Markov's principle is equivalent to allowing instances of the excluded middle to be used as axioms if and only if the conclusion of the $\lor$-elimination rule is a $\Sigma_1^0$ formula. In one sense this tells us that when conclusions are restricted to be $\Sigma_1^0$, allowing premises of arbitrary complexity does not allow us to prove more than what we could prove already with simply existential premises. 

\section{A new negative translation}
\label{sec:new-negat-transl}
\subsubsection{Negative translations}
\label{sec:negat-transl}
Negative translations have been known for long time as a tool to embed classical reasoning into intuitionistic logic. Essentially, they consist in a method to transform every formula provable in a classical theory in another formula that is equivalent in classical logic and that, although it is not intuitionistically equivalent, is provable from the translated  theory. The most prominent example is probably the so called \emph{G\"odel-Gentzen} translation \cite{Gödel33}, or also \emph{double negation} translation. It assigns to every formula $F$ a formula $F^N$ defined by induction on its structure:
\begin{itemize}
\item If $F$ is atomic, $F^N = \neg \neg \ F$
\item $(F_1 \land F_2)^N$ is $F_1^N \land F_2^N$
\item $(F_1 \lor F_2)^N$ is $\neg (\neg F_1^N \land \neg F_2^N)$
\item $(F_1 \to F_2)^N$ is $F_1^N \to F_2^N$
\item $(\neg F)^N$ is $\neg F^N$
\item $(\forall x \ F)^N$ is $\forall x \ \ F^N$
\item $(\exists x \ F)^N$ is $\neg \forall x \ \neg \ F^N$
\end{itemize}

The following theorem states the result we anticipated informally:

\begin{thm}[G\"odel-Gentzen translation]
  \label{thm:negat-transl}
  Let $\Gamma = A_0,\dots A_n$ be a set of formulas. Then $A_1,\dots A_n \vdash A_0$ is classically derivable if and only if $A_1^N,\dots A_n^N \vdash A_0^N$ is intuitionistically derivable.
\end{thm}

For a complete discussion and a proof of this result, one may refer to \cite{Troelstra73}. The translation proves especially useful in the case of arithmetic, thanks to the following theorem

\begin{thm}
\label{thm:neg-trans-ha}
  For any formula $A$ in the language of arithmetic, if $\PA \vdash A$ then $\HA \vdash A^N$
\end{thm}
\begin{proof}
  Thanks to \cref{thm:negat-transl}, we already know that $\PA \vdash A$ if and only if $\HA^N \vdash A^N$. What we need to show is that if $\HA^N \vdash A^N$, then $\HA \vdash A^N$. In order to do so, we need to prove the translated axioms in $\HA$. We know that $\HA \vdash ((s=t) \to \neg \neg (s=t)) \land (\neg \neg (s=t) \to (s=t))$, and therefore since the axioms for equality only use $\forall$ and $\to$, their translation is easily equivalent to the original axiom.

Consider the translation of an instance of the induction axiom:
\[(\forall x (F(x) \to F (\mathbf{s}x)) \to F (0) \to \forall x F(x))^N = \forall x (F^N(x) \to F^N (\mathbf{s}x)) \to F^N (0) \to \forall x F^N(x)\]
Since the second formula is just the instance of the axiom of induction for the formula $F^N$, it is provable in $\HA$. Therefore, we can conclude that $\HA \vdash \HA^N$, and thus $\HA \vdash A^N$
\end{proof}

The negative translation allows to embed all of classical arithmetic inside intuitionistic arithmetic. However, the resulting statements often do not provide a clear computational interpretation: consider for example the translation of an existential statement: we obtain something of the form $\neg \forall x \ \neg F$, and it is not clear how one could exctract a witness. For addressing this issues, one needs another translation such as the A-translation of Friedman \cite{Friedman78}. Essentially, it consists in replacing every atomic predicate $\emp{}$ with $\emp{} \lor A$ for an arbitrary formula $A$. When we combine it with the G\"odel translation, we obtain the following definition: given formulas $F_1,F_2,A$, where no free variable of $A$ is quantified in $F_1$ or $F_2$

\begin{itemize}
\item $\neg_A F_1 = F_1 \to A$, $\bot^A = A$
\item $F_1^A = \neg_A \neg_A F_1$ if $F_1$ is atomic
\item $(F_1 \land F_2)^A = F_1^A \land F_2^A$
\item $(F_1 \lor F_2)^A = \neg_A(\neg_A F_1^A \land \neg_A F_2^A)$
\item $(F_1 \to F_2)^A = F_1^A \to F_2^A$
\item $(\forall x \ F_1)^A = \forall x \ F_1^A$
\item $(\exists x \ F_1)^A = \neg_A \forall x \ \neg_A F_1^A$
\end{itemize}

We can see that it behaves very similarly to the usual G\"odel-Gentzen translation, but with the addition that negation is parametrized by the formula $A$. With thechniques very similar to those of \cref{thm:negat-transl} and \ref{thm:neg-trans-ha} we have that if $\Gamma \vdash F$ in $\PA$, then $\Gamma^A \vdash F^A$ in $\HA$. However, the new translation also allows for a major result for the constructive interpretation of some statements of classical arithmetic:

\begin{thm}[Friedman]
  Let $\emp{}$ be an atomic predicate. Then $\PA \vdash \exists x \ \emp{}(x)$ if and only if $\HA \vdash  \exists x \ \emp{}(x)$
\end{thm}
\begin{proof}
  Suppose $\PA \vdash \exists x \ \emp{}(x)$; then $\HA \vdash (\exists x \ \emp{}(x))^A$, i.e. $\HA \vdash (\forall x \ \neg_A \neg_A \neg_A \emp{}(x)) \to A$. Since it can be seen that $\neg_A \neg_A \neg_A F \dashv \vdash \neg_A F$ in $\HA$ for all $F$, $\HA \vdash (\forall x \ \neg_A \emp{}(x)) \to A$. Now, since we can use any formula for $A$, we use $\exists x \ \emp{}(x)$: in this way we get $\HA \vdash \forall x \ (\emp{}(x) \to \exists x \ \emp{}(x)) \to \exists x \ \emp{}(x)$. Since the antecedent of the formula is provable, we get $\HA \vdash \exists x \ \emp{}(x)$.
\end{proof}

\subsubsection{The $\exists$-translation}
We will now introduce a new translation and consider it for statements of arithmetic. Like the usual negative translations, it will have the property that translated formulas are classically equivalent to the original ones, and that the translated axioms of arithmetic are intuitionistically provable in $\HA$. However, we will not immediately present a result linking classical provability and intuitionistic provability as we did before; indeed, the synctactic translation method presented here will be used in the next section together with a more proof-theoretic technique in order to provide a new interpretation of the simply existential statements of classical arithmetic. 

Our translation is particularly simple when compared with the usual ones. It leaves all logical connectives untouched, except for the case of $\forall$, which is substituted by $\neg \exists \neg$. Formally, we define the translation $\cdot^\exists$ by induction on the structure of the formula:

\begin{itemize}
\item If $F$ is atomic, $F^\exists = F$
\item $(F_1 \land F_2)^\exists$ is $F_1^\exists \land F_2^\exists$
\item $(F_1 \lor F_2)^\exists$ is $F_1^\exists \lor F_2^\exists$
\item $(F_1 \to F_2)^\exists$ is $F_1^\exists \to F_2^\exists$
\item $(\forall x \ F)^\exists$ is $\neg \exists x \ \neg F^\exists$
\item $(\exists x \ F)^\exists$ is $\exists x  \ F^\exists$
\end{itemize}

We know that $\forall x \ A(x)$ is classically equivalent to $\neg \exists x \ \neg A(x)$ regardless of $A$, and thus we can easily state that $\PA \vdash \PA^\exists$ and $\PA^\exists \vdash \PA$. So it is also easy to see

\begin{proposition}
  \label{prop:pa-to-pa-e}
  $\PA \vdash F$ if and only if $\PA^\exists \vdash F^\exists$
\end{proposition}
\begin{proof}
  By a straightforward induction on the derivation.
\end{proof}

The question is a bit more complicated for intuitionistic arithmetic: in general, the translated formula is not intuitionistically equivalent to the original one. Nevertheless, we have the following result:

\begin{thm}
  $\HA \vdash \HA^\exists$. So, every formula provable in $\HA^\exists$ is provable in $\HA$
  \label{thm:exists-translation}
\end{thm}
\begin{proof}
  The axioms for equality and the definition of the successor are left untouched by the translation. Consider now the translation of the axiom for induction for an arbitrary formula $P$:

$$ (Ind)^\exists = (P(0) \land (\forall \alpha \ (P(\alpha)  \to P(\alpha+1))) \to \forall \alpha \ P(\alpha))^\exists = $$
$$P(0) \land (\neg  \exists \alpha. \ \neg (P(\alpha)  \to P(\alpha+1))) \to \neg \exists \alpha \neg P(\alpha)$$
\noindent
The formal derivation in \cref{fig:ind-exists} gives a proof of this formula in $\HA$. Therefore, we have that $\HA \vdash \HA^\exists$, and so also whenever $\HA^\exists \vdash F$ $\HA \vdash F$
\end{proof}

\begin{sidewaysfigure}
  \small{
    \begin{prooftree}
      \AxiomC{$[p]_{(1)} : P(0) \land (\neg \exists \alpha. \ \neg
        (P(\alpha) \to P(\alpha+1)))$}
      \UnaryInfC{$\pi_1(p) : \neg \exists \alpha. \ \neg (P(\alpha)
        \to P(\alpha+1))$}
      
      \AxiomC{$[v]_{(5)} : \neg \neg P (\alpha)$} 

      \AxiomC{$[u]_{(6)} : \neg P(\alpha + 1)$} 
      \AxiomC{$[y]_{(7)} : P(\alpha) \to P(\alpha+1)$} \AxiomC{$[z]_8
        : P(\alpha)$} \BinaryInfC{$yz : P(\alpha+1)$}

      \BinaryInfC{$u(yz) : \bot$} \RightLabel{(8)} \UnaryInfC{$\lambda
        z. u(yz) : \neg P(\alpha)$}

      \BinaryInfC{$v(\lambda z . u(yz)) : \bot$} \RightLabel{$(7)$}
      \UnaryInfC{$ \lambda y \ v(\lambda z . u(yz)) : \neg (P(\alpha)
        \to P(\alpha+1))$}

      \UnaryInfC{$ \langle \alpha, \lambda y \ v(\lambda z . u(yz))
        \rangle : \exists \alpha. \ \neg (P(\alpha) \to P(\alpha+1))$}

      \BinaryInfC{$\pi_1(p)( \langle \alpha, \lambda y \ v(\lambda z
        . u(yz)) \rangle) : \bot$} \RightLabel{$(6)$}
      \UnaryInfC{$ \lambda u. \pi_1(p)( \langle \alpha, \lambda y \
        v(\lambda z . u(yz)) \rangle) : \neg \neg P(\alpha + 1)$}
      \RightLabel{$(5)$}
      \UnaryInfC{$ \lambda v \lambda u. \pi_1(p)( \langle \alpha,
        \lambda y \ v(\lambda z . u(yz)) \rangle) : \neg \neg P
        (\alpha) \to \neg \neg P(\alpha + 1)$}
      \UnaryInfC{$ \lambda \alpha \lambda v \lambda u. \pi_1(p)(
        \langle \alpha, \lambda y \ v(\lambda z . u(yz)) \rangle) :
        \forall \alpha (\neg \neg P (\alpha) \to \neg \neg P(\alpha +
        1))$}
    \end{prooftree}
    \caption{Proof of the inductive step ($\mathcal{D}_1$)}
    \begin{prooftree}
      \AxiomC{$[q]_{(2)} : \exists \alpha \neg P(\alpha)$}

      \AxiomC{$[x]_{(4)} : \neg P(0)$}

      \AxiomC{$[p]_{(1)} : P(0) \land (\neg \exists \alpha. \ \neg
        (P(\alpha) \to P(\alpha+1)))$} \UnaryInfC{$\pi_0(p) : P(0)$}

      \BinaryInfC{$ x \pi_0(p) : \bot$}

      \RightLabel{$(4)$}
      \UnaryInfC{$ \lambda x.x \pi_0(p) : \neg \neg P(0)$}

      \AxiomC{$\mathcal{D}_1$} \noLine
      \UnaryInfC{$ \lambda \alpha \lambda v \lambda u. \pi_1(p)(
        \langle \alpha, \lambda y \ v(\lambda z . u(yz)) \rangle) :
        \forall \alpha (\neg \neg P (\alpha) \to \neg \neg P(\alpha +
        1))$}
      \insertBetweenHyps{\hskip -40pt}
      \RightLabel{Ind}
      \BinaryInfC{$\lambda \alpha \mathbf{R} (\alpha \ \lambda x.x
        \pi_0(p) \ \lambda \alpha \lambda v \lambda u. \pi_1(p)(
        \langle \alpha, \lambda y \ v(\lambda z . u(yz)) \rangle)) :
        \forall \alpha \neg \neg P(\alpha)$}

      \UnaryInfC{$\lambda \alpha \mathbf{R} (\alpha \ \lambda x.x \pi_0(p) \ \lambda
        \alpha \lambda v \lambda u. \pi_1(p)( \langle \alpha, \lambda
        y \ v(\lambda z . u(yz)) \rangle)) \beta: \neg \neg P(\beta)$}
      \AxiomC{$[t]_{(3) (\exists)} : \neg P(\beta)$}
      \insertBetweenHyps{\hskip -100pt}
      \BinaryInfC{$(\lambda \alpha \ (\mathbf{R} (\alpha \ \lambda x.x
        \pi_0(p) \ \lambda \alpha \lambda v \lambda u. \pi_1(p)(
        \langle \alpha, \lambda y \ v(\lambda z . u(yz))
        \rangle)))\beta ) t : \bot$}

      \insertBetweenHyps{\hskip -120pt}

      \RightLabel{(3) ($\exists$-E)}
      \BinaryInfC{$q [(\beta,t).((\lambda \alpha \ (\mathbf{R} (\alpha
        \ \lambda x.x \pi_0(p) \ \lambda \alpha \lambda v \lambda
        u. \pi_1(p)( \langle \alpha, \lambda y \ v(\lambda z . u(yz))
        \rangle))) \beta) t)] : \bot$} \RightLabel{(2)}
      \UnaryInfC{$\lambda q \ q [(\beta,t).((\lambda \alpha \
        (\mathbf{R} (\alpha \ \lambda x.x \pi_0(p) \ \lambda \alpha
        \lambda v \lambda u. \pi_1(p)( \langle \alpha, \lambda y \
        v(\lambda z . u(yz)) \rangle))) \beta) t)] : \neg \exists
        \alpha \neg P(\alpha)$} \RightLabel{(1)}
      \UnaryInfC{$\lambda p \lambda q \ q [(\beta,t).((\lambda \alpha
        \ (\mathbf{R} (\alpha \ \lambda x.x \pi_0(p) \ \lambda \alpha
        \lambda v \lambda u. \pi_1(p)( \langle \alpha, \lambda y \
        v(\lambda z . u(yz)) \rangle))) \beta) t)] : P(0) \land (\neg
        \exists \alpha. \ \neg (P(\alpha) \to P(\alpha+1))) \to \neg
        \exists \alpha \neg P(\alpha) $}
    \end{prooftree}
  }
  \label{fig:ind-exists}
  \caption{Proof of $(Ind)^\exists$ in $\HA$}
\end{sidewaysfigure}

\section{Embedding classical proofs in $\HA + \EM^-$}

We go back now to the system $\HA+\EM^-$ defined in \cref{sec:full-excluded-middle}. Since in this new system instances of the excluded middle are allowed on arbitrary formulas, we might be tempted to investigate more on how much of a classical proof we can reconstruct in it. A first approach can be the following: in the case the statement to be proved is itself simply existential, we could allow occurrences of the excluded middle rule whenever we are sure they are the lowermost infecences. More formally, we introduce the notation

\begin{prooftree}
  \AxiomC{$\mathcal{D}_1$}
  \noLine
  \UnaryInfC{$\exists \alpha \emp{}$}
  \AxiomC{$\mathcal{D}_2$}
  \noLine
  \UnaryInfC{$\exists \alpha \emp{}$}
  \AxiomC{\dots}
  \AxiomC{$\mathcal{D}_n$}
  \noLine
  \UnaryInfC{$\exists \alpha \emp{}$}
  \doubleLine
  \RightLabel{\emmeno}
  \QuaternaryInfC{$\exists \alpha \emp{}$}
\end{prooftree}

to indicate that $\mathcal{D}_1$, $\mathcal{D}_2$ \dots $\mathcal{D}_n$ are proofs of $\exists \alpha \emp{}$ not using \emmeno, possibly with open assumptions, and the conclusion is obtained by repeated usage of the \emmeno rule on them (note that \emmeno is indeed used only on a $\Sigma_1^0$ formula). Similarly define the same notation for \emme.

Then clearly the new construct for \emmeno can be directly replaced by instances of Markov's principle using the proof tree from \cref{sec:full-excluded-middle}. Our task for this section is thus to show that any proof (in $\PA$, i.e. $\HA+\EM$) of a simply existential statement can be rewritten into a proof in $\HA+\EM^-$ of the above form.

In order to do so, we employ new permutation rules extending the ones defined in \cite{Aschieri16} to move the use of classical reasoning below purely intuitionistic proofs. In general, we could have an unrestricted use of the excluded middle, in the form of the rule $\emme$. For every intuitionistic rule, one needs to move the classical rule below it:

$\to$-introduction:
$$\begin{aligned}
\AxiomC{$[A]$}
\AxiomC{$[B]_{(1)}$}
\noLine
\BinaryInfC{$\vdots$}
\noLine
\UnaryInfC{$C$}
\AxiomC{$[\neg A]$}
\AxiomC{$[B]_{(1)}$}
\noLine
\BinaryInfC{$\vdots$}
\noLine
\UnaryInfC{$C$}
\RightLabel{\emme}
\BinaryInfC{$C$}
\RightLabel{(1)}
\UnaryInfC{$B\to C$}
\DisplayProof
  &
  \mbox{ $\leadsto$ } 
  &
\AxiomC{$[A]$}
\AxiomC{$[B]_{(1)}$}
\noLine
\BinaryInfC{$\vdots$}
\noLine
\UnaryInfC{$C$}
\RightLabel{(1)}
\UnaryInfC{$B \to C$}
\AxiomC{$[\neg A]$}
\AxiomC{$[B]_{(2)}$}
\noLine
\BinaryInfC{$\vdots$}
\noLine
\UnaryInfC{$C$}
\RightLabel{(2)}
\UnaryInfC{$B \to C$}
\RightLabel{\emme}
\BinaryInfC{$B\to C$}
\DisplayProof
 \\
\end{aligned}
$$

$\to$-elimination/1:
$$\begin{aligned}
\AxiomC{$[A]$}
\noLine
\UnaryInfC{$\vdots$}
\noLine
\UnaryInfC{$B \to C$}
\AxiomC{$[\neg A]$}
\noLine
\UnaryInfC{$\vdots$}
\noLine
\UnaryInfC{$B \to C$}
\RightLabel{\emme}
\BinaryInfC{$B\to C$}
\AxiomC{$\vdots$}
\noLine
\UnaryInfC{$B$}
\BinaryInfC{$C$}
\DisplayProof
  &
  \qquad \mbox{$\leadsto$ } \qquad
  &
\AxiomC{$[A]$}
\noLine
\UnaryInfC{$\vdots$}
\noLine
\UnaryInfC{$B \to C$}
\AxiomC{$\vdots$}
\noLine
\UnaryInfC{$B$}
\BinaryInfC{$C$}
\AxiomC{$[\neg A]$}
\noLine
\UnaryInfC{$\vdots$}
\noLine
\UnaryInfC{$B \to C$}
\AxiomC{$\vdots$}
\noLine
\UnaryInfC{$B$}
\BinaryInfC{$C$}
\RightLabel{\emme}
\BinaryInfC{$C$}
\DisplayProof
 \\
\end{aligned}
$$

$\to$-elimination/2:
$$\begin{aligned}
\AxiomC{$\vdots$}
\noLine
\UnaryInfC{$B \to C$}
\AxiomC{$[A]$}
\noLine
\UnaryInfC{$\vdots$}
\noLine
\UnaryInfC{$B$}
\AxiomC{$[\neg A]$}
\noLine
\UnaryInfC{$\vdots$}
\noLine
\UnaryInfC{$B$}
\RightLabel{\emme}
\BinaryInfC{$B$}
\BinaryInfC{$C$}
\DisplayProof
  &
  \qquad \mbox{$\leadsto$ } \qquad
  &
\AxiomC{$\vdots$}
\noLine
\UnaryInfC{$B \to C$}
\AxiomC{$[A]$}
\noLine
\UnaryInfC{$\vdots$}
\noLine
\UnaryInfC{$B$}
\BinaryInfC{$C$}
\AxiomC{$\vdots$}
\noLine
\UnaryInfC{$B \to C$}
\AxiomC{$[\neg A]$}
\noLine
\UnaryInfC{$\vdots$}
\noLine
\UnaryInfC{$B$}
\BinaryInfC{$C$}
\RightLabel{\emme}
\BinaryInfC{$C$}
\DisplayProof
 \\
\end{aligned}
$$

$\land$-introduction/1:
$$\begin{aligned}
\AxiomC{$[A]$}
\noLine
\UnaryInfC{$\vdots$}
\noLine
\UnaryInfC{$B$}
\AxiomC{$[\neg A]$}
\noLine
\UnaryInfC{$\vdots$}
\noLine
\UnaryInfC{$B$}
\RightLabel{\emme}
\BinaryInfC{$B$}
\AxiomC{$\vdots$}
\noLine
\UnaryInfC{$C$}
\BinaryInfC{$B \land C$}
\DisplayProof
  &
  \qquad \mbox{$\leadsto$ } \qquad
  &
\AxiomC{$[A]$}
\noLine
\UnaryInfC{$\vdots$}
\noLine
\UnaryInfC{$B$}
\AxiomC{$\vdots$}
\noLine
\UnaryInfC{$C$}
\BinaryInfC{$B \land C$}
\AxiomC{$[\neg A]$}
\noLine
\UnaryInfC{$\vdots$}
\noLine
\UnaryInfC{$B$}
\AxiomC{$\vdots$}
\noLine
\UnaryInfC{$C$}
\BinaryInfC{$B \land C$}
\RightLabel{\emme}
\BinaryInfC{$B \land C$}
\DisplayProof
 \\
\end{aligned}
$$

$\land$-introduction/2:
$$\begin{aligned}
\AxiomC{$\vdots$}
\noLine
\UnaryInfC{$B$}
\AxiomC{$[A]$}
\noLine
\UnaryInfC{$\vdots$}
\noLine
\UnaryInfC{$C$}
\AxiomC{$[\neg A]$}
\noLine
\UnaryInfC{$\vdots$}
\noLine
\UnaryInfC{$C$}
\RightLabel{\emme}
\BinaryInfC{$C$}
\BinaryInfC{$B \land C$}
\DisplayProof
  &
  \qquad \mbox{$\leadsto$ } \qquad
  &
\AxiomC{$\vdots$}
\noLine
\UnaryInfC{$B$}
\AxiomC{$[A]$}
\noLine
\UnaryInfC{$\vdots$}
\noLine
\UnaryInfC{$C$}
\BinaryInfC{$B \land C$}
\AxiomC{$\vdots$}
\noLine
\UnaryInfC{$B$}
\AxiomC{$[\neg A]$}
\noLine
\UnaryInfC{$\vdots$}
\noLine
\UnaryInfC{$C$}
\BinaryInfC{$B \land C$}
\RightLabel{\emme}
\BinaryInfC{$B \land C$}
\DisplayProof
 \\
\end{aligned}
$$

$\land$-elimination/1, $\land$-elimination/2:
$$\begin{aligned}
\AxiomC{$[A]$}
\noLine
\UnaryInfC{$\vdots$}
\noLine
\UnaryInfC{$A_1 \land A_2$}
\AxiomC{$[\neg A]$}
\noLine
\UnaryInfC{$\vdots$}
\noLine
\UnaryInfC{$A_1 \land A_2$}
\RightLabel{\emme}
\BinaryInfC{$A_1 \land A_2$}
\UnaryInfC{$A_i$}
\DisplayProof
  &
  \qquad \mbox{$\leadsto$ } \qquad
  &
\AxiomC{$[A]$}
\noLine
\UnaryInfC{$\vdots$}
\noLine
\UnaryInfC{$A_1 \land A_2$}
\UnaryInfC{$A_i$}
\AxiomC{$[\neg A]$}
\noLine
\UnaryInfC{$\vdots$}
\noLine
\UnaryInfC{$A_1 \land A_2$}
\UnaryInfC{$A_i$}
\RightLabel{\emme}
\BinaryInfC{$A_i$}
\DisplayProof
 \\
\end{aligned}
$$

And similarly for $\lor$-introduction, $\lor$-elimination.

$\exists$-introduction:
$$\begin{aligned}
\AxiomC{$[A]$}
\noLine
\UnaryInfC{$\vdots$}
\noLine
\UnaryInfC{$B[m/\alpha]$}
\AxiomC{$[\neg A]$}
\noLine
\UnaryInfC{$\vdots$}
\noLine
\UnaryInfC{$B[m/\alpha]$}
\RightLabel{\emme}
\BinaryInfC{$B[m/\alpha]$}
\UnaryInfC{$\exists \alpha B$}
\DisplayProof
  &
  \mbox{ $\leadsto$ } 
  &
\AxiomC{$[A]$}
\noLine
\UnaryInfC{$\vdots$}
\noLine
\UnaryInfC{$B[m/\alpha]$}
\UnaryInfC{$\exists \alpha B$}
\AxiomC{$[\neg A]$}
\noLine
\UnaryInfC{$\vdots$}
\noLine
\UnaryInfC{$B[m/\alpha]$}
\UnaryInfC{$\exists \alpha B$}
\RightLabel{\emme}
\BinaryInfC{$\exists \alpha B$}
\DisplayProof
 \\
\end{aligned}
$$

$\exists$-elimination/1:
$$\begin{aligned}
\AxiomC{$[A]$}
\noLine
\UnaryInfC{$\vdots$}
\noLine
\UnaryInfC{$\exists \alpha B$}
\AxiomC{$[\neg A]$}
\noLine
\UnaryInfC{$\vdots$}
\noLine
\UnaryInfC{$\exists \alpha B$}
\RightLabel{\emme}
\BinaryInfC{$\exists \alpha B$}
\AxiomC{$[B]$}
\noLine
\UnaryInfC{$\vdots$}
\noLine
\UnaryInfC{$C$}
\BinaryInfC{$C$}
\DisplayProof
  &
  \qquad \mbox{$\leadsto$ } \qquad
  &
\AxiomC{$[A]$}
\noLine
\UnaryInfC{$\vdots$}
\noLine
\UnaryInfC{$\exists \alpha B$}
\AxiomC{$[B]$}
\noLine
\UnaryInfC{$\vdots$}
\noLine
\UnaryInfC{$C$}
\BinaryInfC{$C$}
\AxiomC{$[\neg A]$}
\noLine
\UnaryInfC{$\vdots$}
\noLine
\UnaryInfC{$\exists \alpha B$}
\AxiomC{$[B]$}
\noLine
\UnaryInfC{$\vdots$}
\noLine
\UnaryInfC{$C$}
\BinaryInfC{$C$}
\RightLabel{\emme}
\BinaryInfC{$C$}
\DisplayProof
 \\
\end{aligned}
$$

$\exists$-elimination/2:
$$\begin{aligned}
\centerfloat
\AxiomC{\vdots}
\noLine
\UnaryInfC{$\exists \alpha B$}
\AxiomC{$[A]$}
\AxiomC{$[B]$}
\noLine
\BinaryInfC{$\vdots$}
\noLine
\UnaryInfC{$C$}
\AxiomC{$[\neg A]$}
\AxiomC{$[B]$}
\noLine
\BinaryInfC{$\vdots$}
\noLine
\UnaryInfC{$C$}
\RightLabel{\emme}
\BinaryInfC{$C$}
\BinaryInfC{$C$}
\DisplayProof
  &
  \mbox{$\leadsto$ }
  &
\AxiomC{\vdots}
\noLine
\UnaryInfC{$\exists \alpha B$}
\AxiomC{$[A]$}
\AxiomC{$[B]$}
\noLine
\BinaryInfC{$\vdots$}
\noLine
\UnaryInfC{$C$}
\BinaryInfC{$C$}
\AxiomC{\vdots}
\noLine
\UnaryInfC{$\exists \alpha B$}
\AxiomC{$[\neg A]$}
\AxiomC{$[B]$}
\noLine
\BinaryInfC{$\vdots$}
\noLine
\UnaryInfC{$C$}
\BinaryInfC{$C$}
\RightLabel{\emme}
\BinaryInfC{$C$}
\DisplayProof
 \\
\end{aligned}
$$

Now, we would like to define a permutation for the case of the universal quantifier. However, it turns out that this is not possible: for the case of $\forall$-I we have no general way of defining one. Consider for example the proof

\begin{prooftree}
  \AxiomC{$[P(x)]_\textsc{em}$}
  \UnaryInfC{$(P(x) \lor \neg P(x))$}
  \AxiomC{$[\neg P(x)]_\textsc{em}$} 
  \UnaryInfC{$(P(x) \lor \neg P(x))$}
  \RightLabel{\emme}
  \BinaryInfC{$(P(x) \lor \neg P(x))$}
  \RightLabel{$\forall$-I}
  \UnaryInfC{$\forall x \ (P(x) \lor \neg P(x))$}
\end{prooftree}

\noindent
Here clearly we have no way of moving the the excluded middle below universal introduction, since the variable $x$ is free before \emme lets us discharge the assumptions. This is where the translation from \cref{sec:new-negat-transl} comes to the rescue: clearly, proofs in $\PA^\exists$ will not contain applications of rules for the universal quantifier, and are thus suitable for our transformations. Therefore, the last rule for which we should give a permutation is the translated rule of induction $(Ind)^\exists$ for $\PA^\exists$:

\begin{prooftree}
  \AxiomC{$\Gamma \vdash A(0)$}
  \AxiomC{$\Gamma \vdash \neg \exists \alpha \neg \ (A(\alpha) \to A(\suc(\alpha)))$}
  \RightLabel{$Ind^\exists$}
  \BinaryInfC{$\Gamma \vdash \neg \exists \alpha \ \neg A(\alpha)$}
\end{prooftree}

The permutations for $Ind^\exists$ will be:

\begin{prooftree}
  \AxiomC{\vdots}
  \noLine
  \UnaryInfC{$B(0)$}
  \AxiomC{[$A$]} 
  \noLine
  \UnaryInfC{$\vdots$}
  \noLine
  \UnaryInfC{$\neg \exists \alpha \ \neg B((\alpha)\rightarrow B(\suc(\alpha)))$} 
  \AxiomC{[$\neg A$]} 
  \noLine
  \UnaryInfC{$\vdots$} 
  \noLine
  \UnaryInfC{$\neg \exists \alpha \ \neg (B(\alpha)\rightarrow B(\suc(\alpha)))$}
  \RightLabel{\emme}
  \BinaryInfC{$\neg \exists \alpha \ \neg (B(\alpha)\rightarrow B(\suc(\alpha)))$} 
  \BinaryInfC{$\neg \exists \alpha \ \neg B$}
\end{prooftree}

converts to:

\begin{prooftree}
  \AxiomC{\vdots} \noLine \UnaryInfC{$B(0)$}
  \AxiomC{[$A$]} \noLine \UnaryInfC{$\vdots$} \noLine
  \UnaryInfC{$\neg \exists \alpha \ \neg (B(\alpha)\rightarrow B(\suc(\alpha)))$}
  \BinaryInfC{$\neg \exists \alpha \ \neg B$}
  \AxiomC{\vdots} \noLine \UnaryInfC{$B(0)$}
  \AxiomC{[$\neg A$]} \noLine \UnaryInfC{$\vdots$}
  \noLine
  \UnaryInfC{$\neg \exists \alpha \ \neg (B(\alpha) \rightarrow B(\suc(\alpha)))$}
  \BinaryInfC{$\neg \exists \alpha \ \neg B$}
  \RightLabel{\emme}
  \BinaryInfC{$\neg \exists \alpha \ \neg B$}
\end{prooftree}

and

\begin{prooftree}
  \AxiomC{[$A$]} 
  \noLine
  \UnaryInfC{$\vdots$}
  \noLine
  \UnaryInfC{$B(0)$}
  \AxiomC{[$\neg A$]} 
  \noLine
  \UnaryInfC{$\vdots$} 
  \noLine
  \UnaryInfC{$B(0)$}
  \RightLabel{\emme}
  \BinaryInfC{$B(0)$} 
  \AxiomC{\vdots}
  \noLine
  \UnaryInfC{$\neg \exists \alpha \ \neg (B(\alpha)\rightarrow B(\suc(\alpha)))$} 
  \BinaryInfC{$\neg \exists \alpha \ \neg B$}
\end{prooftree}

converts to:

\begin{prooftree}
  \AxiomC{[$A$]} 
  \noLine
  \UnaryInfC{$\vdots$}
  \noLine
  \UnaryInfC{$B(0)$}
  \AxiomC{\vdots}
  \noLine
  \UnaryInfC{$\neg \exists \alpha \ \neg (B(\alpha)\rightarrow B(\suc(\alpha)))$} 
  \BinaryInfC{$\neg \exists \alpha \ \neg B$}
  \AxiomC{[$\neg A$]} 
  \noLine
  \UnaryInfC{$\vdots$} 
  \noLine
  \UnaryInfC{$B(0)$}
  \AxiomC{\vdots}
  \noLine
  \UnaryInfC{$\neg \exists \alpha \ \neg (B(\alpha)\rightarrow B(\suc(\alpha)))$} 
  \BinaryInfC{$\neg \exists \alpha \ \neg B$}
  \RightLabel{\emme}
  \BinaryInfC{$\neg \exists \alpha \ \neg B$}
\end{prooftree}
\noindent
By employing the just defined permutation rules, we can state

\begin{proposition}
  \label{lemma:normal-pa-exists}
  Every proof of a formula $F$ in $\PA^\exists$ can be transformed into a proof
\begin{prooftree}
  \AxiomC{$\mathcal{D}_1$}
  \noLine
  \UnaryInfC{$F$}
  \AxiomC{$\mathcal{D}_2$}
  \noLine
  \UnaryInfC{$F$}
  \AxiomC{\dots}
  \AxiomC{$\mathcal{D}_n$}
  \noLine
  \UnaryInfC{$F$}
  \doubleLine
  \RightLabel{\emme}
  \QuaternaryInfC{ $F$}
\end{prooftree}

Where $\mathcal{D}_1, \mathcal{D}_2 \dots \mathcal{D}_n$ are purely intuitionistic proofs.
\end{proposition}

\begin{proof}
 Proceed by induction on the structure of the proof. The base case where the proof only containts axioms and a single rule is vacuous. Otherwise, assume there is at least one use of \emme (if not the thesis holds vacuosly) and consider the lowermost rule application:
\begin{itemize}
\item If it is \emme, then the induction hypothesis can be applied to the subtrees corresponding to the two premises of the rule, yelding the thesis.
\item As an example to the case of unary rules, consider $\exists$-introduction; then the proof has the shape 
\AxiomC{\vdots}
\noLine
\UnaryInfC{$F^\prime[m/\alpha]$}
\RightLabel{$\exists$-I}
\UnaryInfC{$\exists \alpha F^\prime$}
\DisplayProof

Applying the induction hypothesis to the subproof corresponding to the premise, by our assumption we get a proof of the form

\begin{prooftree}
  \AxiomC{$\mathcal{D}_1$}
  \noLine
  \UnaryInfC{$F^\prime[m/\alpha]$}
  \AxiomC{$\mathcal{D}_2$}
  \noLine
  \UnaryInfC{$F^\prime[m/\alpha]$}
  \AxiomC{\dots}
  \AxiomC{$\mathcal{D}_n$}
  \noLine
  \UnaryInfC{$F^\prime[m/\alpha]$}
  \doubleLine
  \RightLabel{\emme}
  \QuaternaryInfC{$F^\prime[m/\alpha]$}
\end{prooftree}

Substitute this in the original proof: by applying the permutation rule for $\exists$-introduction $n-1$ times, we move the exist introduction right below the intuitionistic part; the proof then becomes

\begin{prooftree}
  \AxiomC{$\mathcal{D}_1$}
    \noLine
  \UnaryInfC{$F^\prime[m/\alpha]$}
  \UnaryInfC{$\exists \alpha F^\prime$}
  \AxiomC{$\mathcal{D}_2$}
    \noLine
  \UnaryInfC{$F^\prime[m/\alpha]$}
  \UnaryInfC{$\exists \alpha F^\prime$}
  \AxiomC{\dots}
  \AxiomC{$\mathcal{D}_n$}
    \noLine
  \UnaryInfC{$F^\prime[m/\alpha]$}
  \UnaryInfC{$\exists \alpha F^\prime$}
  \doubleLine
  \RightLabel{\emme}
  \QuaternaryInfC{$\exists \alpha F^\prime$}
\end{prooftree}

Which satisfies the thesis.

The cases of the other unary rules are analogous.

\item As an example for the case of binary rules, consider $\to$-elimination; then the proof has the shape 
\AxiomC{\vdots}
\noLine
\UnaryInfC{$G \to F$}
\AxiomC{\vdots}
\noLine
\UnaryInfC{$G$}
\RightLabel{$\to$-E}
\BinaryInfC{$F$}
\DisplayProof

Applying the induction hypothesis to the subproofs corresponding to the premises, by our assumption we get two proofs where in at least one of the two the last used rule is \emme: we select one where this is the case.

Say we chose the proof of the first premise (the other case is symmetric), then from the induction hypothesis we have obtained a proof of the shape
\begin{prooftree}
  \AxiomC{$\mathcal{D}_1$}
  \noLine
  \UnaryInfC{$G \to F$}
  \AxiomC{$\mathcal{D}_2$}
  \noLine
  \UnaryInfC{$G \to F$}
  \AxiomC{\dots}
  \AxiomC{$\mathcal{D}_n$}
  \noLine
  \UnaryInfC{$G \to F$}
  \doubleLine
  \RightLabel{\emme}
  \QuaternaryInfC{$G \to F$}
\end{prooftree}

After substitutig this in the original proof, we can employ the permutation for $\to$-elimination $n-1$ times and obtain the proof

\begin{prooftree}
  \AxiomC{$\mathcal{D}_1$}
  \noLine
  \UnaryInfC{$G \to F$}
  \AxiomC{\vdots}
  \noLine
  \UnaryInfC{$G$}
  \BinaryInfC{$F$}
  \AxiomC{$\mathcal{D}_2$}
  \noLine
  \UnaryInfC{$G \to F$}
  \AxiomC{\vdots}
  \noLine
  \UnaryInfC{$G$}
  \BinaryInfC{$F$}
  \AxiomC{\dots}
  \AxiomC{$\mathcal{D}_n$}
  \noLine
  \UnaryInfC{$G \to F$}
  \AxiomC{\vdots}
  \noLine
  \UnaryInfC{$G$}
  \BinaryInfC{$F$}
  \doubleLine
  \RightLabel{\emme}
  \QuaternaryInfC{$F$}
\end{prooftree}

If the proof of $G$ is intuitionistic we have the thesis, so assume it is not. Just as before, we can use the induction hypothesis on it, and obtain:
\begin{prooftree}
  \AxiomC{$\mathcal{D}_1$}
  \noLine
  \UnaryInfC{$G \to F$}

  \AxiomC{$\mathcal{E}_1$}
  \noLine
  \UnaryInfC{$G$}
  \AxiomC{\dots}
  \AxiomC{$\mathcal{E}_m$}
  \noLine
  \UnaryInfC{$G$}
  \doubleLine
  \RightLabel{\emme}
  \TrinaryInfC{$G$}

  \BinaryInfC{$F$}

  \AxiomC{\dots}

  \AxiomC{$\mathcal{D}_n$}
  \noLine
  \UnaryInfC{$G \to F$}
  \AxiomC{$\mathcal{E}_1$}
  \noLine
  \UnaryInfC{$G$}
  \AxiomC{\dots}
  \AxiomC{$\mathcal{E}_m$}
  \noLine
  \UnaryInfC{$G$}
  \doubleLine
  \RightLabel{\emme}
  \TrinaryInfC{$G$}
  \BinaryInfC{$F$}
  \doubleLine
  \RightLabel{\emme}
  \TrinaryInfC{$F$}
\end{prooftree}

After applying $m-1$ times the second permutation for $\to$-elimination, we obtain

\begin{prooftree}
  \centerfloat
  \AxiomC{$\mathcal{D}_1$}
  \noLine
  \UnaryInfC{$G \to F$}
  \AxiomC{$\mathcal{E}_1$}
  \noLine
  \UnaryInfC{$G$}
  \BinaryInfC{$F$}

  \AxiomC{\dots}

  \AxiomC{$\mathcal{D}_1$}
  \noLine
  \UnaryInfC{$G \to F$}
  \AxiomC{$\mathcal{E}_m$}
  \noLine
  \UnaryInfC{$G$}
  \BinaryInfC{$F$}

  \doubleLine
  \RightLabel{\emme}
  \TrinaryInfC{$F$}

  \AxiomC{\dots}

  \AxiomC{$\mathcal{D}_n$}
  \noLine
  \UnaryInfC{$G \to F$}
  \AxiomC{$\mathcal{E}_1$}
  \noLine
  \UnaryInfC{$G$}
  \BinaryInfC{$F$}

  \AxiomC{\dots}

  \AxiomC{$\mathcal{D}_n$}
  \noLine
  \UnaryInfC{$G \to F$}
  \AxiomC{$\mathcal{E}_m$}
  \noLine
  \UnaryInfC{$G$}
  \BinaryInfC{$F$}

  \doubleLine
  \RightLabel{\emme}
  \TrinaryInfC{$F$}

  \doubleLine
  \RightLabel{\emme}
  \TrinaryInfC{$F$}
\end{prooftree}

Which satisfies the thesis. The cases of the other binary rules are analogous.

\end{itemize}
\end{proof}

After these transformations we are using the excluded middle only with the statement to prove as a conclusion. A similar result was obtained by Seldin \cite{Seldin89}, but with a rule for reduction ad absurdum in place of the excluded middle and without induction.

This means that if the statement we are proving is of a certain complexity, we do not need classical reasoning on formulas of higher complexity.

\begin{proposition}
  \label{lemma:pa-ha-exists}
  Every proof in $\PA^\exists$ of a $\Sigma_1^0$ statement can be transformed into a proof in $\HA^\exists+\EM_1^-$
\end{proposition}

\begin{proof}
  By the \cref{lemma:normal-pa-exists} we know we can transform any proof in $\PA^\exists$ of a statement $\exists \alpha \emp{}$ into a proof of the form

\begin{prooftree}
  \AxiomC{$\mathcal{D}_1$}
  \noLine
  \UnaryInfC{$\exists \alpha \emp{}$}
  \AxiomC{$\mathcal{D}_2$}
  \noLine
  \UnaryInfC{$\exists \alpha \emp{}$}
  \AxiomC{\dots}
  \AxiomC{$\mathcal{D}_n$}
  \noLine
  \UnaryInfC{$\exists \alpha \emp{}$}
  \doubleLine
  \RightLabel{\emme}
  \QuaternaryInfC{$\exists \alpha \emp{}$}
\end{prooftree}

Since every application of \emme happens on a simply existential statement, we can directly replace them with \emmeno. Moreover, from \cref{sec:full-excluded-middle} we know that \emmeno is equivalent to \emeno, and thus we obtain a proof in $\HA^\exists + \EM_1^-$ as desired.
\end{proof}

Finally, we can conclude the section with the main theorem

\begin{thm}
  If $\PA \vdash \exists x \ \emp{}$ with $\emp{}$ atomic, then $\HA + \EM_1^- \vdash \exists x \ \emp{}$
\end{thm}
\begin{proof}
  Given a proof of $\exists x \ \emp{}$ in $\PA$, by \cref{prop:pa-to-pa-e} we can apply the $\exists$-translation and obtain a proof of $(\exists x \ \emp{})^\exists = \exists x \ \emp{}$ in $\PA^\exists$. Then, by \cref{lemma:pa-ha-exists}, we can transform this in a proof in $\HA^\exists+\EM_1^-$. Finally, thanks to \cref{thm:exists-translation}, we know that $\HA +\EM_1^- \vdash \exists x \ \emp{}$.
\end{proof}


\chapter{Conclusions}
\label{cha:conclusions}
We have seen how the interpretation of Markov's principle in constructive settings has been an historically controversial matter. Kreisel showed by means of the modified realizability that it could not be validated by intuitionistic logic, while Kleene's realizability semantics, although successful in interpreting it, reduced it to a mere unbounded search and thus brought it back to non-constructivity. However, we noted how a much more refined interpretation of Markov's principle was already present in G\"odel's work; W.W. Tait, in a more recent analysis of G\"odel's position on intuitionism, notices how in more modern terms we could state that ``if one is looking for methods of proof which automatically yield algorithms for computing a witness for existential theorems, intuitionistic logic is too narrow'' \cite{Tait06}.

Following more recent lines of research, we introduced the logic $\HA + \EM_1$, a related term calculus following the propositions as types paradigm, and a realizability interpretation of the former into the latter. We saw how the term system of $\HA+\EM_1$ provides a tool to investigate the computational content of Markov's principle, and we interpreted the principle as learning program that gets a witness for the conclusion supposing the assumption does not hold and repeatedly testing it. Moreover, we introduced a restricted version of $\HA + \EM_1$ called $\HA + \EM_1^-$, where we showed that the new logical principle added to the logic is equivalent to Markov's principle. This new system inherits the Curry-Howard correspondence and the realizability interpretation from the one it derives from.

By means of $\HA+\EM_1^-$, we have obtained a new proof of constructivity of intuitionistic arithmetic extended with Markov's principle. Finally, we have generalized the obtained result and shown that Markov's principle is also equivalent to adding to intuitionistic arithmetic the principle $\mathsf{EM}^-$, that is a restricted form of the rule of excluded middle where we are only allowed to use it in disjunction eliminations if the conclusion is simply existential.

This final observation led us to the introduction of a new negative translation in the style of the classic ones by G\"odel and Friedman. Our new translation has the advantage of not changing $\exists$ and $\lor$ when compared to the usual ones, but needs a series of permutation rules to be applied on proofs in order to be useful. Combining these results, we obtained a way to transform classical proofs of $\Sigma_1^0$ statements into proofs in $\HA+\EM_1^-$, thus allowing us to extract programs from any of these proofs.

Our system $\HA+\EM_1^-$ is reminescent of the one of Herbelin \cite{Herbelin10}. Here, a deductive system for intuitionistic logic is extended with the two rules \textsc{throw} and \textsc{catch} and is equipped with a Curry-Howard correspondence:
$$\begin{array}{c} \Gamma \vdash_{\alpha: T, \Delta} p: T\\
 \hline \Gamma\vdash_{\Delta} \texttt{catch}_\alpha p : T
 \end{array}\ \textsc{catch}
 \ \ \ \
\begin{array}{c} \Gamma \vdash_{\Delta} p: T \ \ \ \ (\alpha : T) \in \Delta\\
 \hline \Gamma\vdash_{\Delta} \texttt{throw}_\alpha p : C
 \end{array}\ \textsc{throw}
$$

The reduction rules for the lambda terms \texttt{catch} and \texttt{throw} define a mechanism of delimited exceptions. Herbelin addresses pure first order logic, and obtains for Markov's principle the term

\[ \lambda a. \mathtt{catch}_\alpha \mathtt{efq} \ a (\lambda b. \mathtt{throw}_\alpha b) : \neg \neg T \to T \]
 where $T$ is a $\forall$, $\to$-free formula. The behaviour of the term is similar to the one presented in \cref{sec:real-mark-princ}. Thanks to this, Herbelin proves the constructivity of the logic by showing the disjunctive and existential properties.

However, in this work we have related Markov's principle to the other semi-classical principle $\EM_1^-$, and thus the logical part of the system results much clearer. In addition, we extended a system of arithmetic whereas Herbelin's work addressed pure intuitionistic logic; by presenting also a realizability interpretation, the extracted programs can be interpreted as a way to actually compute the witnesses for existential statements. 

Another related work is \cite{Aschieri14}. Here the authors extend modified realizability, and thus the work has the advantage of using a purely functional language. However, just like modified realizability, the realizability interpretation that is provided does not satisfy subject reduction and is therefore not suitable for the investigation of the logical properties of the system. Moreover, the realizer for Markov's principle is

\[ \lambda z^{(\Nat \to U) \to U}\langle \mathtt{quote} (z \mathtt{mtest}_{\lambda x. P}), \mathtt{if }\ P^\bot [ \mathtt{quote} (z \mathtt{mtest}_{\lambda x . P})/x] \ \mathtt{then} \ \mathtt{tt}0 \ \mathtt{else} \ z \ \mathtt{mtest}_{\lambda x.P} \rangle \]

\noindent
where $\mathtt{mtest}_{\lambda x. P} := \lambda x^\Nat.\mathtt{if}\ P \ \mathtt{then}\ \mathtt{tt}x \ \mathtt{else}\ \mathtt{tt}x$. This is much less clear than what we have seen so far, and relies on an internal comunication system based on the primitive type $U$, terms $\top_0 : U, \top_1 : U, \dots$ and $\bot_0 : U, \bot_1 : U, \dots$ and the reduction rules
\begin{align*}
   &\mathtt{tt} n \mapsto \top_n & \mathtt{ff}n \mapsto \bot_n \\
   &\mathtt{quote} \top_m \mapsto m & \mathtt{quote} \bot_m \mapsto m 
\end{align*}

As mentioned, this thesis has the advantages of a clearer explanation of Markov's principle and of presenting a system that enjoys subject reduction.

\subsubsection{Future Work}
As known for example from the field of proof mining \cite{Kohlenbach08} Markov's principle is fundamental for the aim of extracting constructive information from non purely constructive proofs. A similarly important principle is the \emph{double negation shift}, stated as
\[ \forall x \neg \neg A(x) \to \neg \neg \forall x A(x) \]

for $A$ atomic. In \cite{Ilik12}, Danko Ilik showed that an intuitionistic logic extended with this principle retains the disjunctive and existential properties, using techniques similar to those of Herbelin. In the same work, he mentions that Herbelin had also extended his calculus of delimited control operators to a system proving this principle. Given the relation between Herbelin's work on Markov's principle and our current work, it is interesting to see if one could develop a modified version of $\HA + \EM^-$ that is able to interpret the double negation shift. A candidate could be the system $\IL +\EM$ presented in \cite{Aschieri16}: we conjecture that a version of this system for arithmetic with restrictions similar to those presented in this thesis would be constructive; its relationship with other principles remains to be studied.



\backmatter
\nocite{*}
\printbibliography

\end{document}